\documentclass[10pt, letterpaper, twocolumn]{IEEEtran} 

\usepackage{amsmath}
\usepackage{amssymb}    
\usepackage{amsfonts}
\usepackage[english]{babel}
\usepackage{cite} 
\usepackage{multirow}
\usepackage{stfloats}    

\usepackage{algorithm}   
\usepackage{algorithmic} 

\usepackage{caption} 
\usepackage{graphicx}
\usepackage{epsfig}
\usepackage{subfigure} 
\usepackage{tablefootnote}

\usepackage{diagbox}   

\usepackage{url}

\usepackage{accents}
\makeatletter
\def\widebar{\accentset{{\cc@style\underline{\mskip10mu}}}}
\def\Widebar{\accentset{{\cc@style\underline{\mskip13mu}}}}
\makeatother

\usepackage{footnote}
\usepackage{pifont}
\usepackage[perpage]{footmisc}  

\newtheorem{theorem}{Theorem}

\newtheorem{definition}{Definition}
\newtheorem{proposition}{Proposition}

\usepackage{booktabs}    
\allowdisplaybreaks  

 \usepackage[draft]{changes}   
	\definechangesauthor[name={Yunquan}, color=orange]{hcs}   
	\setremarkmarkup{(#2)}

\begin{document}

\captionsetup[figure]{labelfont={ }, name={Fig.}, labelsep=period} 
\pagestyle{empty}

\title{Energy Harvesting Powered Sensing in IoT: Timeliness Versus Distortion}
\author{Yunquan~Dong,~\IEEEmembership{Member,~IEEE}, 
                 Pingyi Fan~\IEEEmembership{Senior Member,~IEEE},
                 and~Khaled~Ben~Letaief,~\IEEEmembership{Fellow,~IEEE}
\vspace{-8mm}

        \thanks{Y. Dong is with the School of Electronic and Information Engineering,  Nanjing University of Information Science and Technology, Nanjing 210044, China (e-mail: yunquandong@nuist.edu.cn).

        P. Fan is with the Department of Electronic Engineering, Tsinghua University, Beijing 100084, China (e-mail:  fpy@tsinghua.edu.cn).

        Khaled B. Letaief is with the Department of Electrical and Computer Engineering, HKUST, Clear Water Bay, Kowloon, Hong Kong (e-mail: eekhaled@ust.hk).

        This work was supported by the National Natural Science Foundation of China (NSFC) under Grant 61701247.
        }
        }


\maketitle
\thispagestyle{empty}

\begin{abstract}
We consider an Internet-of-Things (IoT) system in which an energy harvesting powered sensor node monitors the phenomenon of interest and transmits its observations to a remote monitor over a Gaussian channel.
    We measure the timeliness of the signals recovered by the monitor using age of information (AoI), which could be reduced by transmitting more observations to the monitor.
We evaluate the corresponding distortion with the mean-squared error (MSE) metric, which would be reduced if a larger transmit power and a larger source coding rate were used.
    Since the energy harvested by the sensor node is random and limited, however, the timeliness and the distortion of the received signals cannot be optimized at the same time.
Thus, we shall investigate the timeliness-distortion trade-off of the system by minimizing the average weighted-sum AoI and distortion over all possible transmit powers and transmission intervals.
    First, we explicitly present the optimal transmit powers for the performance limit achieving save-and-transmit policy and the easy-implementing fixed power transmission policy.
Second, we propose a backward water-filling based offline power allocation algorithm and a genetic based offline algorithm to jointly optimize the transmission interval and transmit power.
    Third, we formulate the online power control as an Markov Decision Process (MDP) and solve the problem with an iterative algorithm, which closely approach the trade-off limit of the system.
Also, we show that the optimal transmit power is a monotonic and bi-valued function of current AoI and distortion.
    Finally,  we present our results via numerical simulations and extend results on the save-and-transmit policy to fading sensing systems.


\end{abstract}

\begin{keywords}
Age of information, Internet of Things, sensing systems, energy harvesting.
\end{keywords}

\section{Introduction}
\IEEEPARstart {W}{ith} the rapid development in embedded systems, multi-terminal communications, and cloud computing, more and more smart devices and low-power sensors are connected to the internet, which is referred to as the Internet-of-Things (IoT).
    In particular, IoT networks have been increasingly popular in scenarios related to data gathering and service sharing in recent years, e.g.,  environment monitoring and smart city planning~\cite{Monitoring-IoTJ-2018}, industrial automation \cite{industrial-auto}, and target surveillance and tracking \cite{uav-tracking}.
 In these systems, a number of sensor nodes are used to monitor the phenomenon of interest constantly and to report the obtained observations to a remote center in real-time.
    Different from traditional communication systems in which data rate (or throughput) is the most important metric, the distortion and the timeliness of the recovered signals are more concerned in IoT based monitoring systems.
 That is, we are more interested in whether the signal recovered by the monitor can precisely characterize the phenomenon and whether the signal is timely or outdated.
     When the system supports a higher data rate, the monitor only sees a reduction in distortion or an improvement in the timeliness of information transmission.

In IoT networks and sensor neworks, the distortion of the recovered signals are often measured by the mean-squared error (MSE) of the decoded signal or the estimated signal. 
    In \cite{dong-TWC-correlated_sensing}, for example, the weighted-sum distortion in recovering two correlated Gaussian sources was optimized in the framework of network information theory.
In \cite{Dong-TSP-2019, Xiao-Linear-2008}, the random source is estimated by combining sensor observations with a best linear unbiased estimator (BLUE), for sensor networks with orthogonal channels and coherent multiple access channels (MAC), respectively.

 The timeliness of the recovered signals can be measured by the age of information (AoI) metric, which is defined as the elapsed time (i.e., the age) after the generation of the latest received observation \cite{Vnet-1-2011}.
    In particular, minimizing the AoI of the system can guarantee the timeliness of sensing while increasing the throughput or decreasing the transmission delay cannot.
For example, the recovered signals would be outdated if the throughput is so large that the delay waiting for being transmitted is large or the throughput is so small that few new recoveries are available at the monitor (no matter the delay is small or not).
    Therefore, AoI has been widely studied in various real-time applications, e.g., in sensor-based monitoring systems \cite{ Gu-2019-mornitoring, Niu-2019-RR1}, health monitoring systems \cite{health-Proc-2012}, cognitive radio-based IoT systems \cite{Gu-2019-cognitive}, and two-way data exchanging systems \cite{Dong-2018-infocom, Dong-2019-access, Dong-2019-jcn}.

 On the other hand, the limitation in the energy supply of sensor nodes puts formidable challenges to IoT networks.
    To be specific, the battery capacity of  sensor nodes is generally small due to device size constraints and cost considerations, which significantly limits the life-time of sensors.
 In view of this, the energy harvesting technology was developed \cite{ulukus-JSAC-eh}.
    By utilizing an energy harvesting unit and an energy buffer, sensor nodes can harvest energy (e.g., solar
and wind power) from the ambient environment, thereby ensuring an unlimited energy supply for each node.
    However, the arrivals of energy are sporadic and irregular.
To better utilize the harvested energy, therefore, we need to schedule the usage of energy carefully.
    First, if the harvesting process is fully predictable (i.e., known non-causally at transmitter), the harvested energy can
be scheduled in an offline manner \cite{Uysal-2015-ita, Ulukus-TWC-2019}.
    In this scenario, the scheduling of energy usage turns to be deterministic and can be solved before the transmission process actually begins.
Second, if the energy harvesting process cannot be well predicted, the online energy scheduling is required, in which each node adjusts its transmit power based on the previous and current energy states in real-time \cite{ Uysal-2015-ita, Ulukus-TWC-2019,Uysal-2017-ISIT, Yangjing-2017-TGCN,Zhou-TVT-2016}.
    In particular, online energy schedulings can be readily solved through the Markov decision process (MDP) based stochastic control, the Lyapunov optimization technique, or the semi-definite relaxation framework.

In energy harvesting powered IoT networks, however, the distortion and the timeliness (i.e., the average AoI) of the recovered signals can not be optimized at the same time.
    On one hand, sensor nodes should transmit more observations to reduce the average AoI, which inevitably reduces the corresponding transmit power, since the available energy is limited.
As a results, the source coding rate of the observations has to be reduced, and thus the distortion of the recovered signals would be increased.
    On the other hand, a smaller distortion can be achieved if the sensor uses longer time to accumulate energy so that a large transmit power can be used.
Along with the reduction in distortion, however, the average AoI would definitely be increased.

In this paper, therefore, we shall approach the best timeliness-distortion trade-off by minimizing the average weighted-sum AoI and distortion of the system.
    Specifically, we consider a monitoring system with an energy harvesting powered sensor node and a remote monitor.
The sensor observes the phenomenon and transmits its observations to the monitor over a Gaussian channel when it has saved the required energy.
    We first investigate the timeliness-distortion limits of the system through a save-and-transmit policy, in which the sensor saves all the harvested energy in the energy buffer for a long time and then transmits observations with a fixed transmit power and a fixed transmission interval.
We then investigate the performance of the fixed power transmission policy, the offline and the online power control for the system.
    The obtained results on the save-and-transmit policy and a fixed power transmission policy will also be extended to sensing systems with block Rayleigh fading.
    The main contribution of the paper can be summarized as follows.

\begin{itemize}
    \item We present the timeliness-distortion limit of the system by studying the performance of the save-and-transmit policy.
                    We also show that the fixed power transmission policy is a simple yet well behaved scheme.
    \item We propose a backward water-filling based offline power control scheme for the system with a given sequence of transmission intervals.
            With this power control scheme, we then propose a genetic based algorithm to jointly optimize the transmission intervals and the corresponding transmit powers.
    \item We model the online power control of the system as an MDP and solve the problem by an iterative algorithm.
        We show that the optimal transmit power is non-decreasing with the energy state of the sensor and is a bi-valued function of the current AoI and distortion.
\end{itemize}

\subsection{Organizations}

This paper is organized as follows.
   Section~\ref{sec:2_model} presents the sensing model, energy harvesting model, the definition of AoI, and the formulation of our problem.
 We investigate the timeliness-distortion performance of the fixed power transmission and the save-and-transmit policy in Section~\ref{sec:3_fx_sv}.
    In Section \ref{sec:4_offline}, we study the oflline power control of a system with a finite period of observations and transmissions.
In Section \ref{sec:5_online}, we discuss the online power control through an MDP formulation.
    In particular, we present the monotonicity of the cost function and the optimal transmit power, with respect to the AoI, the distortion, and the energy state, respectively.
The theoretic results are confirmed via numerical simulations Section~\ref{sec:6_simulation}.
    Finally,  we conclude the paper in Section~\ref{sec:5_conclusion}.


%

\begin{figure}[!t]
\centering
\includegraphics[width=3.0in]{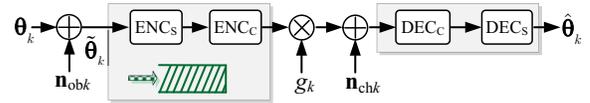}
\caption{The sensing system. $\text{ENC}_\text{S}$ denotes the source encoder,  $\text{ENC}_\text{C}$ is the channel encoder, $\text{DEC}_\text{S}$ is the source decoder, and $\text{DEC}_\text{C}$ is the channel decoder.} \label{fig:system_model}
\end{figure}

\section{System Model}\label{sec:2_model}
Consider a sensing system with a remote monitor and an energy harvesting powered sensor, which observes a certain phenomenon (source signal $\theta$) characterized by a Gaussian process from time to time.
    The observations will be encoded with lossy source coding and be transmitted to the monitor over a Gaussian channel, as shown in Fig. \ref{fig:system_model}.
Finally, the monitor will decode the received signal and restore the source signal under some distortion constraints.

\subsection{Sensing Model}
We assume that time is discrete and the slot length is $T_{\text{s}}$.
    Each block $T_{\text{B}}$ consists of $2J$ slots and can be either an \textit{idle block} or a \textit{busy block}.
That is, the sensor can choose to stay idle or to make $J$ observations $\{\widetilde{{\theta}}_{kj}, j=1,2,\cdots,J\}$ in the first $J$ slots and transmit them to the remote monitor in the remaining slots, where the blocks are indexed by $k$ and the slots are indexed by $j$.
    Due to accuracy issues, the observations suffer from independently and identically distributed (i.i.d.) Gaussian observation noises ${ n}_{\text{ob}kj}$ with zero-mean and variance $\sigma^2_{\text{ob}}$, i.e., $\widetilde{{\theta}}_{kj}=\theta_{kj}+{n}_{\text{ob}kj}$.

In the $k$-th block, we denote the transmit power of the sensor as $P_k$ and the power gain of the channel as $g_k$.
    We also denote the channel noise power as $\sigma^2_{\text{ch}}$, and the frequency bandwidth of transmitted signals as $W$.
    Moreover, we consider a set of following assumptions.

\begin{itemize}
  \item [A1] $\{\theta_{kj}, k= 1,2,\cdots, j=1,2,\cdots,J\}$ is a stationary Gaussian process with zero-mean and variance $\sigma^2_\theta$.
  \item [A2] $P_k$ is the normalized transmit power taking values from non-negative integers, i.e., $P_k=0,1,2,\cdots$.
  \item [A3] $g_k=1$ for all $k\geq1$, i.e., the link between the sensor and the monitor is a channel with additive white Gaussian noise (AWGN).
  \item [A4] $T_{\text{s}}=\frac1{2 W}$, i.e., $T_{\text{s}}$ equals to the maximum sampling time.
\end{itemize}

 We note that the results obtained based on Assumption A1 can also be generalized to systems with non-stationary sources  \cite{Reily-ComST-2014} or quasi-stationary sources \cite{Joda-TCom-2013}, as long as the period before the next change in distribution is sufficiently long for the required source coding and channel coding (e.g., we have $10^5$ channel uses per second with slot length $T_{\text{s}}=10~\mu$s).
    In Assumption A2, each $P_k$ is obtained from $P_k=P'_k/\widebar P$, in which $P'_k$ is the actual transmit power and $\widebar P$ is a normalizing factor.
In particular, we incorporate $\widebar P$ into the channel noise. That is, for each given noise power $\sigma^2_0$, we shall rewrite the channel noise power as $\sigma^2_{\text{ch}}=\widebar P\sigma^2_0$.
    Although $P_k$ is an integer,  the channel SNR can be treated as continuous since $\widebar P$ can be any small positive numbers as desired.
On Assumption A3, it is noted that the effect of channel attenuation can be considered by varying the channel noise power.
    Given the noise power $\sigma^2_0$ and normalizing factor $\widebar P$,  for example, we can set $\sigma^2_{\text{ch}}= \widebar P\sigma^2_0 \iota^a$ and change $\sigma^2_{\text{ch}}$ instead of varying $\iota$, in which $\iota$ is the sensor-monitor distance and $a$ is the pathloss exponent.
   By considering some additional multiplicative channel gains, this model can also be extended to systems with block fading channel (cf. Subsection \ref{subsec:6_B}).

 Note that the instantaneous capacity of the sensor-monitor channel is given by~\cite[Chap.~9.1, \textit{Theorem} 9.1.1]{Cover-IT-Book}
\begin{equation}\label{eq:ch_rate}
    r_{\text{ch}k} = W \log\left( 1+\frac{P_k}{\sigma^2_{\text{ch}}} \right).
\end{equation}

In each busy block, the obtained sequence of $J$ observations will be encoded into an index $m_k$ using lossy source coding \cite[Chap.~10.2, \textit{Definition} 10.7]{Cover-IT-Book}.
    Afterwards, $m_k$ will be encoded into an ideal channel codeword,  which will be transmitted to the monitor
    in the following $J$ slots.
Upon receiving a distorted channel codeword, the monitor will perform channel decoding and source decoding sequentially to obtain a sequence of restored signals $\{\widehat{\theta}_{kj}, j=1,2,\cdots,J\}$ with a certain quantization distortion $\sigma^2_{\text{qu}k}=\mathbb{E}[(\widetilde{\theta}_{kj}-\widehat{\theta}_{kj})^2]$.
    According to \cite[Chap.~9.1, \textit{Theorem} 9.1.1]{Cover-IT-Book}, the minimum source coding rate for each sample is given by
\begin{equation}\label{eq:qu_rate}
    r_{\text{sc}k}  = \frac12 \log \frac{\sigma^2_\theta+\sigma^2_{\text{ob}}}{\sigma^2_{\text{qu}k}}.
\end{equation}

By combining \eqref{eq:ch_rate}, \eqref{eq:qu_rate}, Assumption A3, and $r_{\text{ch}k}JT_{\text{s}} = r_{\text{sc}k}J$, we have
\begin{equation} \label{eq:sigma_qu_1}
    \sigma^2_{\text{qu}k} = \frac{\sigma^2_\theta+\sigma^2_{\text{ob}} }
                {1+\frac{P_k}{\sigma^2_{\text{ch}}}}.
\end{equation}

It is also shown in~\cite[Chap.~10.3, Fig. 10.5]{Cover-IT-Book} that the rate-distortion limit approaching source coding from observation $\widetilde{\theta}_{kj}$ to recovery $\widehat{\theta}_{kj}$ can be characterized by the following test channel
    \begin{equation} \label{eq:x_first}
        \widetilde{\theta}_{kj} = \widehat{\theta}_{kj} + n_{\text{qu}kj},
    \end{equation}
    where $n_{\text{qu}kj}$ is the i.i.d. Gaussian quantization noise with zero mean and variance $\sigma^2_{\text{qu}k}$ and $\widehat{\theta}_{kj}$ is randomly generated according to a certain optimal Gaussian distribution.
In particular, $n_{\text{qu}kj}$ and $\widehat{\theta}_{kj}$ are independent from each other.

    Since a noisy observation can also be expressed as $\widetilde{\theta}_{kj}= \theta_{kj}+n_{\text{ob}kj}$, the recovery $\widehat{\theta}_{kj}$ can further be written as
    \begin{equation}\label{eq:model}
      \widehat{\theta}_{kj} = \theta_{kj} - n,
    \end{equation}
    where $n = {n}_{\text{qu}kj} - {n}_{\text{ob}kj}$ is the total noise.

According to \cite[\textit{Proposition} 1]{Dong-TSP-2019} and \eqref{eq:sigma_qu_1}, the total noise power, i.e., the distortion of a busy block, is given by
\begin{align}\label{eq:theta_dist}
    D_k =\sigma^2_{\text{ob}} +  \frac{(\sigma^2_\theta-\sigma^2_{\text{ob}})\sigma^2_{\text{ch}}}
                                        {\sigma^2_{\text{ch}}+P_{k}}.
\end{align}

Moreover, the distortion does not change until another busy block is performed and completed.

\subsection{Energy Harvesting and Energy Usage} \label{subsec:pc_policy}
We assume that the sensor is powered by energy harvesting, e.g., by harvesting energy from the wind.
    For notational simplicity, we shall present energy by $\widebar PT_{\text{B}}$ Joule per unit.
In doing so, we can compare power and energy directly.

In each block, we assume that the sensor harvests $E_k=1$ unit of energy with probability $\lambda$.
    That is, we have $\Pr\{E_k=1\}=\lambda$ and  $\Pr\{E_k=0\}=1-\lambda$.
When a unit of energy is harvested, it will be saved in an energy buffer and can be used in future blocks.
    Without loss of generality, we assume that the energy buffer is large and the probability of energy overflow is negligible.

At the beginning of each block, the sensor determines its transmit power $P_k$ (can be zero, i.e., stays idle) according to some power control scheme.
    If $P_k>0$ is used and the remaining energy is no less than $P_k$, the sensor will make a sequence of $J$ observations, encode them and then transmit them to the monitor.
In particular, we consider the following power control schemes:
\begin{itemize}
  \item  \textit{fixed power transmission} in which $P_k=P_{\text{fx}}$ for all $k\geq1$;
  \item  \textit{save-and-transmit policy} in which $K$ goes to infinity so that we can save all the harvested energy for a long time and then transmit with a fixed (also optimal) power $P_{\text{sv}}$ and a fixed transmission interval $X_{\text{sv}}$;
  \item \textit{offline power control} in which the energy harvesting information is available at the sensor non-causally and $P_k$ can be optimally determined before the transmissions actually start;
  \item \textit{online scheduling} in which the energy harvesting information is causally available and $P_k$ is determined based on the current energy state, age of information, and distortion.
\end{itemize}

For each power control scheme, it should be noted that the transmit power suffers from the following causality constraint
\begin{equation}\label{eq:energy_cstr}
    \sum_{i=1}^k P_i \leq \sum_{i=1}^{k-1} E_i, ~~~\forall ~k\geq1.
\end{equation}

\subsection{Age of Information}
Age of information is a measure of information freshness, as defined as below.

\begin{figure}[!t]
\centering
\includegraphics[width=3.1in]{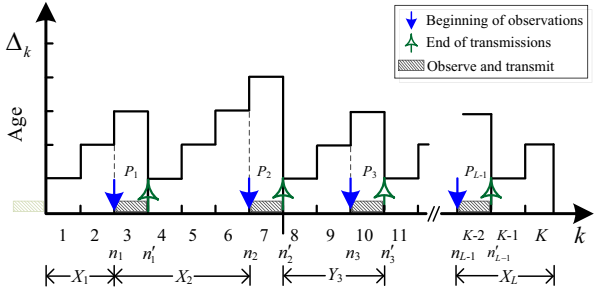}
\caption{Age of information, where $X_{l}=n_l-n_{l-1}$ is the inter-transmission time  and  $Y_{l}=n'_l-n'_{l-1}$ is the inter-departure time.  } \label{fig:aoi}
\end{figure}

\begin{definition}
     (\textit{Age of Information-AoI}~\cite{Yates-2012-age}). At the beginning of the $k$-th block, the index of the most recent busy block is
          \begin{equation} \nonumber
                N_{\text{U}}(k) = \max\{ l|n_l'< k \},
            \end{equation}
in which $n'_l$ is the end of the $l$-th busy block.
      The age of information of the system is then defined as the random process
    \begin{equation} \label{df:aoi}
        \Delta_k = k - N_{\text{U}}(k),
    \end{equation}
    which does not change during each block.

\end{definition}

As shown in Fig. \ref{fig:aoi}, we consider a period of $K$ blocks which consists of $L-1$ busy blocks.
    We denote the time between two transmission beginnings as inter-transmission time $X_l = n_l-n_{l-1}$ and the time between two consecutive transmission completions as inter-departure time $Y_l= n'_l-n'_{l-1}$.
We assume that there is a virtual busy block performed with transmit power $P_0=0$ in block $k=0$.
    Thus, the distortion of inter-departure time $Y_1$ would be $D_0=\sigma^2_\theta$.
Moreover, the energy harvested in $X_L$ will not be used and we denote $Y_L = K-n'_{L-1}$.
    Thus, we have $Y_1=X_1+1$, $Y_L=X_L-1$,  and $Y_l=X_l$ for $2\leq l\leq L-1$.
Furthermore, the AoI returns to $\Delta_l= 1$ in the first block of each inter-departure time $Y_l$ and would be increased by one in the following blocks, until another busy block is completed (e.g., blocks 3, 7, and 10).

\subsection{Problem Formulation}
In this paper, we shall minimize the average of the weighted-sum AoI and distortion by scheduling the energy usage of the sensor node.
Note that the AoI $\Delta_k$ is given by \eqref{df:aoi} and the distortion $D_k$ is given by \eqref{eq:theta_dist}.
    The optimization problem, therefore, can be expressed as
    \begin{align}
\label{prob:0_0}
 (\textbf{P}_0)~~
    \min\limits_{\{P_k\}} ~~~&\frac1K \sum_{k=1}^K (\Delta_k+wD_k) \\
 \label{prob:0_1}
        \text{subject~to}~ & \sum_{i=1}^k P_i \leq \sum_{i=1}^{k-1} E_i,~~~\forall~ k\geq1, \\
            & P_k\geq 0,\qquad\qquad~~~\forall~ k\geq1, \qquad\quad\quad
\end{align}
where $w>0$ is a positive weighting coefficient.

In this paper, \textit{Problem} \textbf{P}$_0$ will be investigated under the fixed power transmission, offline scheduling, online scheduling, and the save-and-transmit policy, respectively.

\section{Long-Term Age-Distortion Trade-off} \label{sec:3_fx_sv}
In this section, the considered period is assumed to be infinitely long (i.e, $K\rightarrow\infty$), and we investigate the age-distortion limit of the sensing system as expressed in \textit{Problem} \textbf{P}$_0$ under the fixed power transmission policy and the save-and-transmit policy.

\subsection{Fixed Power Transmission} \label{subsec:3_fix_p}
In the fixed power transmission, the sensor would stay idle for a few blocks, until it has accumulated enough energy to perform a busy block at a given and fixed transmit power $P_{\text{fx}}\geq1$.
        Thus, the distortion $D_k$ would be constant throughout the period, i.e., $D_k=D_{\text{fx}}$.
    We denote the number of blocks for the sensor to accumulate the required energy for the $l$-th busy block as $\tau_{\text{H}l}$.
As shown in Fig. \ref{fig:aoi}, we have $\tau_{\text{H}l}=X_l$.
    Note that when $P_{\text{fx}}$ is large, $\tau_{\text{H}l}$  would be large,  which leads to  large  AoIs.
Note also that when $P_{\text{fx}}$ is large, the distortion $D_{\text{fx}}$ is small.
    Therefore, there is a natural trade-off between the average AoI and the distortion.
With the fixed power transmission policy, \textit{Problem} \textbf{P}$_0$ can be characterized by the following proposition.
\begin{proposition} \label{prop:fixed}
    With the fixed power transmission, \textit{Problem} \textbf{P}$_0$ reduces to
    \begin{align}
\label{prob:1_0}
  (\textbf{P}_1)~~~
    \min\limits_{P_{\text{fx}}}~~~~&\frac{P_{\text{fx}}+1}{2\lambda} + w\sigma^2_{\text{ob}} +  \frac{w(\sigma^2_\theta-\sigma^2_{\text{ob}})\sigma^2_{\text{ch}}}
                                        {\sigma^2_{\text{ch}}+P_{\text{fx}}} \quad  \\
 \label{prob:1_1}
        \text{subject~to}~~ & P_{\text{fx}}\geq 1,
\end{align}
    for which the optimal transmit power is $P_{\text{fx}}=1$
    if $w\leq w_0$ or $\sigma^2_{\text{ob}}\geq \sigma^2_{\text{ob}0}$, and
    \begin{equation}\label{rt:p_prop1_2}
    P_{\text{fx}}=\sqrt{2\lambda w(\sigma^2_\theta-\sigma^2_{\text{ob}})\sigma^2_{\text{ch}}} - \sigma^2_{\text{ch}}
\end{equation}
    if $w> w_0,\sigma^2_{\text{ob}}< \sigma^2_{\text{ob}0}$, where $w_0$ and $\sigma^2_{\text{ob}0}$ are, respectively, given by
    \begin{align} \label{eq:w0}
        w_0 &= \frac{(1+\sigma^2_{\text{ch}})^2} {2\lambda\sigma^2_\theta\sigma^2_{\text{ch}}}, \\
        \label{eq:r_ob0}
        \sigma^2_{\text{ob}0} &= \sigma^2_\theta\left(1-\frac{w_0}{w}\right).
    \end{align}
\end{proposition}
\begin{proof}
    See Appendix \ref{prf:fixed}.
\end{proof}

In general, the average AoI is much larger than the average distortion (less than unity).
    If $w$ is small, the distortion would contribute little to the objective function, and thus
 it would be wise to transmit more frequently at a small transmit power, e.g., $P_{\text{fx}}=1$.
    It is also noted that the distortion $D_{\text{fx}}$ (cf. \eqref{eq:theta_dist}) is monotonically increasing with $P_{\text{fx}}$ if $\sigma^2_{\text{ob}}> \sigma^2_{\text{ob}0}$.
In this case, it would also be better to transmit at $P_{\text{fx}}=1$, regardless of $w$.
    Moreover, although the  $P_{\text{fx}}$ is used as an integer in the proof (cf. \eqref{apx:a_1}) of the proposition,  the solution given in \eqref{rt:p_prop1_2}  is not necessarily an integer and is optimal for \textit{Problem} \textbf{P}$_1$ in all cases.

\begin{figure}[!t]
\centering
\includegraphics[width=3.7in]{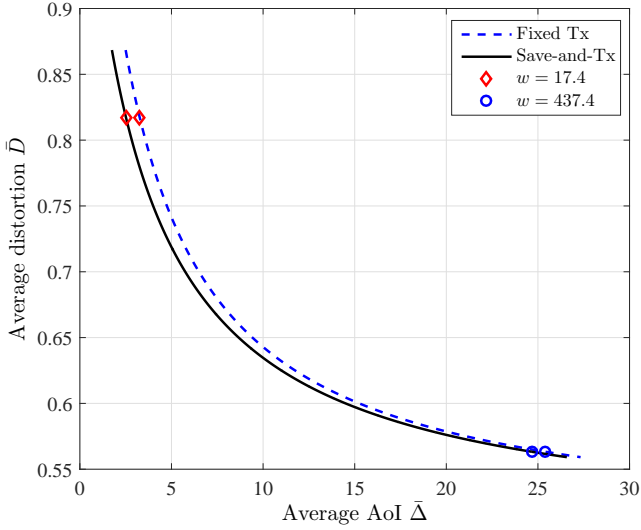}
\caption{AoI-distortion trade-off.  } \label{fig:age_distortion_trade}
\end{figure}

\subsection{Save-and-Transmit Policy}
In the save-and-transmit policy, the sensor stays idle and saves all the harvested energy in its energy buffer for a period of $o(K)$ blocks, in which $o(K)\rightarrow\infty$ and is an infinitesimal of $K$.
    Afterwards, the sensor performs observations and transmissions with a fixed transmit power $P_{\text{sv}}$ and a fixed inter-transmission time $X_{\text{sv}}$.
Similar to \cite{Ulukus-2012-awgn}, we see that the probability for the sensor not having enough energy to perform transmissions once in each $X_{\text{sv}}$ blocks goes to zero as $K\rightarrow\infty$.
    In this part, we do not require the transmit power $P_{\text{sv}}$ to be an integer and have
    \begin{align}
        P_{\text{sv}} &= \lambda X_{\text{sv}}, \\
        D_{\text{sv}} &= \sigma^2_{\text{ob}} +  \frac{(\sigma^2_\theta-\sigma^2_{\text{ob}})\sigma^2_{\text{ch}}}
                                        {\sigma^2_{\text{ch}}+\lambda X_{\text{sv}} }, \\
        \Delta_{\text{sv}} &= \frac12 (Y_{\text{sv}}+1).
    \end{align}

For a given  $w$, we will then optimize $P_{\text{sv}}$ and $X_{\text{sv}}$ by solving the following problem.
    \begin{align}
\label{prob:2_0}
  (\textbf{P}_2)~~~
    \min\limits_{P_{\text{sv}}}~~~~&\frac{P_{\text{sv}}+\lambda}{2\lambda} + w\sigma^2_{\text{ob}} +  \frac{w(\sigma^2_\theta-\sigma^2_{\text{ob}})\sigma^2_{\text{ch}}}
                                        {\sigma^2_{\text{ch}}+ P_{\text{sv}}}   \\
 \label{prob:2_1}
        \text{subject~to}~~ & P_{\text{sv}}\geq\lambda.
\end{align}

It is seen that \textit{Problem} \textbf{P}$_2$ shares a similar solution as  \textit{Problem} \textbf{P}$_1$.
    Thus, $P_{\text{sv}}$ would be given by \eqref{rt:p_prop1_2} and we have $w'_0= (\lambda+\sigma^2_{\text{ch}})^2 /(2\lambda\sigma^2_\theta\sigma^2_{\text{ch}})$,  $\sigma'^2_{\text{ob}0} =\sigma^2_\theta(1-w'_0/w)$.
Howeverver, it should be noted that $\Delta_{\text{sv}}$ is a constant (since $X_{\text{sv}}$ is deterministic) while $\Delta_{\text{fx}l}$ is random (since $X_{\text{fx}l}=\tau_{\text{H}l}$ is random).
    Moreover, the AoI $\Delta_{\text{sv}}$ would be slightly smaller than $\widebar\Delta_{\text{fx}}$ (cf. \eqref{apx:delta_fix}) since $\lambda$ is less than one in general.
More importantly, by using a long energy saving phase, the sensor node can operate as if it has a constant power supply.
    Therefore, the age-distortion trade-off achieved by the save-and-transmit policy defines the performance limit of the sensing system.

Fig. \ref{fig:age_distortion_trade} presents the trade-off between the average AoI and the average distortion.
    We set $\lambda=0.4$, $\sigma^2_\theta=1$, $\sigma^2_{\text{ob}}=0.5$,  $\sigma^2_{\text{ch}}=2.8$, and $w$ taking values from $\Omega=\{w_0:0.5:500\}$.
From, \eqref{eq:w0} and \eqref{eq:r_ob0}, we have $w_0=12.8929$ and $\sigma^2_{\text{ob}0}=0.6777$, and thus the above parameters ensures that $P_{\text{fx}}=P_{\text{sv}}$ are larger than unity.
    As is shown, the achievable average AoI is relatively small while the average distortion is relatively large when $w$ is with a small value (e.g., $w=17.4$), and vise vasa.
It is also observed that the performance of the fixed transmission is lower bounded by that of the save-and-transmit policy.

\section{Offline Power Control} \label{sec:4_offline}
For the offline power control, we do not require the transmit power to be an integer and assume that the energy harvesting process is known at the sensor non-causally.

\subsection{Problem Reformulation}
We consider a period with $K$ blocks and $L-1$ busy blocks.
    We denote the transmit power of the $l$-th busy block as $P_l$ and rewrite the energy harvested in the $j$-th block of $X_l$ as $E_{lj}$.
As is shown in Fig. \ref{fig:aoi}, $P_l$ is determined primarily by the energy harvested during the previous inter-transmission time $X_{l-1}$.
    Moreover, the distortion $D_l$ is determined by $P_l$ as shown in equation \eqref{eq:theta_dist} and does not change during $Y_{l+1}$.
Thus, \textit{Problem} \textbf{P}$_0$ can be rewritten as
\begin{align}
\label{prob:3_0}
 (\textbf{P}_3)
    \min\limits_{\{P_l, Y_l, L\}}  \,&\frac1K \sum_{l=1}^L
                    (\delta_l+wD_{l-1}Y_{l}) \\
 \label{prob:3_1}
        \text{subject~to} \,& \sum_{l=1}^L Y_l  = K, \\
 \label{prob:3_2}
                    & \sum_{i=1}^l P_i \leq \sum_{i=1}^{l}\sum_{j=1}^{X_{l}} E_{ij},~~~\forall~ \,1\leq l\leq L,
\end{align}
    where $D_0=\sigma^2_\theta$ and $\delta_l=(Y_l+Y_l^2)/2$ is the total AoI during the $l$-th inter-departure time (cf. Fig. \ref{fig:aoi}).

However, the optimization over $\{X_l\}$ and $\{P_l\}$ are correlated and \textit{Problem} \textbf{P}$_3$ is not convex in $\{X_l\}$.
    Therefore, we shall solve the problem by using an iterative algorithms.

\subsection{Optimal Transmit Power}
For each given $L$ and each feasible (i.e., \eqref{prob:3_1} is satisfied) sequence of inter-transmission time $\{X_l, l=1,2,\cdots,L\}$,  the inter-departure time $\{Y_l, l=1,2,\cdots,L\}$ and the average AoI of the system can be determined, and thus \textit{Problem} \textbf{P}$_3$ is equivalent to
\begin{align}
\label{prob:3_2_0}
 (\textbf{P}'_3)~~~
    \min\limits_{\{ P_l\}}~  \,&\frac1{K} \sum_{l=1}^{L}  \frac{Y_{l}}
                                        {\sigma^2_{\text{ch}}+P_{l-1}} &
                       \quad \\
 \label{prob:3_2_1}
        \text{subject~to} ~
                    & \sum_{i=1}^l P_i \leq \sum_{i=1}^{l}\sum_{j=1}^{X_{l}} E_{ij}, &\forall~ 1\leq l\leq L-1.
\end{align}

\begin{proposition} \label{prop:3_off p}
    The optimal transmit power $P_l$ for \textit{Problem} \textbf{P}$'_3$ is given by
    \begin{equation}\label{rt:off_opt_p}
        P_l = \max \left(\sqrt{\frac{Y_l}{K \nu_l} } -\sigma^2_{\text{ch}}, e_\varepsilon\right),
    \end{equation}
    where $e_\varepsilon$ is a very small positive valued transmit power to indicate a busy blocks and $\nu$ is the water-level given by
    \begin{equation}
        \nu_l = \sum\nolimits_{j=l}^L \mu_j, ~~l=1,2,\cdots, L-1
        \end{equation}
        and $\mu_j$ is the positive Lagrangian multiplier.
\end{proposition}

    \begin{proof}
        See \textit{Appendix} \ref{prf:3_off p}
    \end{proof}

We note that
    \begin{equation} \label{eq:nul}
        -\nu_l=\frac{-Y_{l+1}}{K(\sigma^2_{\text{ch}}+P_{l})^2}
     \end{equation}
     is equal to the first order derivative of objective function $\mathcal{J}=\frac1{K} \sum_{l=0}^{L-1}  \frac{Y_{l+1}} {\sigma^2_{\text{ch}}+P_{l}}$ with respect to $P_l$.
Thus, $\nu_l$ presents the marginal gain of using more power in the $l$-th busy block and $\mathcal{J}$ would be minimized if each $\nu_l$ is as close to each other as possible (under constraint \eqref{prob:3_2_1}).
    We denote the harvested energy during the $l$-th inter-transmission time $X_l$ as
    \begin{equation}
        e_l=\sum\nolimits_{j=1}^{X_l}E_{lj}.
    \end{equation}

Similar to \cite{JYang-2012-jcn-backward}, the optimal transmit powers can be obtained using a generalized backward water-filling algorithm, as shown in \textit{Algorithm} \ref{alg:wb_wf}.
    To be specific, the algorithm starts from $X_L$ by setting $P_{L-1}=e_{L-1}$ and calculating $\nu_{L-1}$ using \eqref{eq:nul}.
    Next, we move to $X_{L-1}$ and pour $e_{L-2}$ to the first block of  $X_{L-1}$ until $e_{L-2}$ is depleted or $\nu_{L-2}=\nu_{L-1}$.
In the latter case, the remaining energy $e_\text{r}$ will be poured into (the first blocks of) $X_L$ and $X_{L-1}$ in such a way that $\nu_{L-2}$ and $\nu_{L-1}$ remain equal.
   This process stops when the energy $e_1$ harvested in $X_1$ has been used and $P_1$ has been determined.

\begin{algorithm}[!t]
\algsetup{linenosize=\scriptsize}
\scriptsize
\caption{Weighted backward water-filling}
\begin{algorithmic}[1]\label{alg:wb_wf}
\REQUIRE ~~\\
\STATE Set $\Delta e = 10^{-4}$ and $\epsilon= 10^{-5}$;
    \STATE Set $P_l=0, \nu_l=\infty, l=1,2,\cdots,L$;
    \STATE Pour energy $e_{L-1}$ into the first block of $X_L$ and set $P_{L-1}=e_L$;
    \STATE Calculate $\nu_{L-1}$ using \eqref{eq:nul};
\ENSURE ~~\\
\FOR {$l=L-2$ to $1$}
    \STATE Pour energy $e_l$ into $X_{l+1}$ and set $P_l=e_l$;
    \STATE Calculate $\nu_l$ using \eqref{eq:nul};
    \IF {$\nu_l<\nu_{l+1}$}
        \STATE Reset transmit power $P_l$ by \eqref{rt:off_opt_p};
        \STATE Update remaining energy as $e_\text{r}=e_l-P_l$;
        \STATE Update water-level as $\nu_l = \nu_{l+1}$;
        \WHILE {$e_\text{r}>0$}
            \STATE Find index $i=\arg \max_{l\leq j\leq L} \nu_j$;
            \STATE  Reset $P_i=P_i+\Delta e$ and update $\nu_i$ using \eqref{eq:nul};
            \STATE  $e_\text{r}=e_\text{r}-\Delta e$;
        \ENDWHILE
    \ENDIF
\ENDFOR
 \STATE \textbf{Output:} $\{P_l\}$.
\end{algorithmic}
\end{algorithm}

\subsection{Optimal Inter-Transmission Time}
For each given $L$ and each feasible (i.e., \eqref{prob:3_2} is satisfied) sequence of transmit power $\{P_l, l=1,2,\cdots,L-1\}$, the distortion sequence $\{D_l\}$ can be calculated by \eqref{eq:theta_dist} and \textit{Problem} \textbf{P}$_3$ reduces to
\begin{align}
\label{prob:3_1_0}
 (\textbf{P}''_3)~~~
    \min\limits_{\{ X_l\}}~~~  \,&\frac1{2K} \sum_{l=1}^{L}
                     ( Y_l^2 + 2wD_{l-1}Y_l ) \quad \\
 \label{prob:3_1_1}
        \text{subject~to} ~~& -\sum_{i=1}^l X_i  \leq -k_l, ~\forall~1\leq l\leq L,\\
        & X_l \in \mathcal{X}, \qquad\quad~~~\forall~1\leq l\leq L,
\end{align}
where $\mathcal{X}= \{ 1,2, \cdots, K\}$ is the feasible set for the inter-transmission time $X_l$.
    Moreover, \eqref{prob:3_1_1} presents the energy causality of the system and $k_l$ is the first block by which the sensor has harvested enough energy for the first $l$ busy blocks, respectively, with transmit power $P_1, P_2,\cdots$, and $P_l$.

Since we have $Y_1=X_1+1$, $Y_L=X_L-1$,  and $Y_l=X_l$ for $2\leq l\leq L-1$, we shall replace $Y_l$ with $X_l$ in \textit{Problem} \textbf{P}$''_3$.

It is observed that \textit{Problem} \textbf{P}$''_3$ is a separable integer optimization problem and can be solved by the backward dynamic programming \cite[Chapt. 7.1.1]{Li-integeropt-2006}.
    When $L$ is large, however, the backward dynamic programming algorithms is almost impossible to implement due to the `curse of dimensionality'.
Moreover, searching for the optimal number (i.e., $L$) of busy blocks is also computation consuming.
    Therefore, we shall jointly solve \textit{Problems} \textbf{P}$'_3$ and \textbf{P}$''_3$ using a genetic algorithm, as shown in \textit{Algorithm} \ref{alg:genetic}.

We refer to each possible length of an inter-transmission time $X_l$ as a gene-instance and refer to a row vector of $K$ gene-instances as a chromosome.
    When $N_{\text{pop}}$ chromosomes are considered, we denote the chromosomes as $\boldsymbol{x}_i, i=1,2,\cdots,N_{\text{pop}}$ and denote the gene-instances as $x_{ik},  k=1,2,\cdots,K$.
Thus, we have $\boldsymbol{x}_i=[x_{i1}, x_{i2},\cdots,x_{iK}]$ and $x_{ik}\in\mathcal{X}$.
    Note that $\boldsymbol{x}_i$ is not necessarily a feasible solution to \textit{Problem} \textbf{P}$''_3$ since $\sum_{k=1}^{K} x_{ik}$ is most probably larger than $K$.

In the initialization phase, $N_{\text{pop}}$ chromosomes are generated by uniformly drowning each $x_{ik}$ from $\mathcal{X}$.
    In each of the following $N_{\text{iter}}$ iterations, the fitness of the chromosomes is evaluated first.

To be specific, we take as many gene-instances $x_{ik}$ as possible from a chromosome $\boldsymbol{x}_i$ until their sum exceeds $K$ for the first time, i.e., $\sum_{k=1}^{L_i} x_{ik}\geq K$.
    Afterwards, the last gene-instance $x_{iL_i}$ will be updated by removing the excess portion, i.e., $x_{iL_i}=x_{iL_i}- (\sum_{k=1}^{L_i} x_{ik}-K)$.
It is clear that $\{x_{i1},x_{i2}, \cdots, x_{iL_i}\}$ (removing zero elements if any) is a feasible sequence of inter-transmission times and the corresponding optimal transmit power $\{P_{i1},P_{i2}, \cdots, P_{iL_i}\}$ can be obtained by \textit{Algorithm} \ref{alg:wb_wf}.
    Furthermore, the total age $\delta_l$ and the total distortion $D_{l-1}Y_l$ of each inter-departure time, as well as the weighted-sum AoI and distortion (cf. \eqref{prob:3_0}) can be calculated readily.
We define the inverse of the weighted-sum AoI and distortion as the fitness of the chromosome, i.e.,
\begin{align}\label{df:fitness}
    f_i = K \left(\sum_{l=1}^{L_i} (\delta_{il}+wD_{l-1}x_{il})\right)^{-1}.
\end{align}

\begin{algorithm}[!t]
\algsetup{linenosize=\scriptsize}
\scriptsize
\caption{Genetic based joint optimization}
\begin{algorithmic}[1]\label{alg:genetic}
\REQUIRE ~~\\
    \STATE Initialize the energy harvesting process;
    \STATE Initialize the chromosome population by uniformly selecting genes from $\mathcal{X}$;
\ENSURE ~~\\
    \FOR {$l=1$ to $N_{\text{iter}}$}
        \STATE For each chromosome, take the first few genes to form a feasible inter-transmission sequence, calculate the optimal transmit power sequence using \textit{Algorithm} \ref{alg:wb_wf}, calculate the fitness by \eqref{df:fitness};
        \STATE Select $N_{\text{parent}}$ parent chromosomes using selection rate $q_{\text{sel}}$;
        \STATE Generate $N_{\text{pop}}-N_{\text{parent}}$ child chromosomes using crossover depth $D_{\text{cross}}$;
        \STATE Record the average weighted-sum AoI and distortion $J_l$, the inter-transmission time $\{X_l\}$, the transmit power $\{P_l\}$ of the best chromosome;
    \ENDFOR
 \STATE \textbf{Output:}  The $\{X_l\}$ and $\{P_l\}$ with the smallest $J_l$.
\end{algorithmic}
\end{algorithm}

With a normalized probability $q_i=f_i/\big(\sum_{j=1}^{N_{\text{pop}}}f_i)\big)$, each chromosome would be randomly selected and added to a parent-set.
    The selection process does not stop until the number of chromosomes in the set reaches $N_{\text{parent}}=N_{\text{pop}}q_{\text{sel}}$, where $q_{\text{sel}}$ is the selection rate.
    While these parent chromosomes are kept in the whole chromosome set, the unselected ones will be discarded.
Next, $N_{\text{child}}=N_{\text{pop}}-N_{\text{parent}}$ child chromosomes will be generated based on their respective randomly chosen parent chromosomes.
    For a pair of parent chromosomes $\boldsymbol{x}_{\text{p}}$ and $\boldsymbol{x}_{\text{m}}$, for example, two children chromosomes will be produced, respectively, by increasing $D_{\text{cross}}$ randomly chosen gene-instances of $\boldsymbol{x}_{\text{p}}$ by one and reducing $D_{\text{cross}}$ randomly chosen gene-instances of $\boldsymbol{x}_{\text{m}}$ by one, where $1\leq D_{\text{cross}}\leq K$ is referred to as the crossover depth.
If any gene-instance turns to be negative, it will  be reset to zero.
    After the crossover operation, the algorithm goes to the next iteration and finally terminates at the $N_{\text{iter}}$-th iteration.

Since the weighted-sum AoI and distortion of each iteration is not strictly decreasing, we shall record the best chromosome of each iteration and take the best one among them as the final output.

\section{Online Power Control}\label{sec:5_online}

In this section, we investigate the online power control of the system, in which the sensor node knows the energy harvesting process casually and adjusts its transmit power based on the current AoI, distortion, and energy state in real-time.

\subsection{Online Problem Formulation}
Since the energy harvesting rate $\lambda$ is less than unity, the probability for the sensor to have a large amount of energy in the buffer, i.e., $B_k\rightarrow\infty$, is approximately zero.
    Also, the probability for the AoI $\Delta_k$ to be very large is approximately zero, since the energy harvesting rate is strictly positive.
Therefore, it is reasonable to assume that $\Delta_k$ and $B_k$ are upper bounded by $\delta_{\text{max}}$ and $b_{\text{max}}$, respectively.

Motivated by the Markovian structure of the AoI process, we cast the online power control problem as an MDP, as shown below.
\begin{itemize}
    \item \textbf{States:} We define the state of the system as $\boldsymbol{s}=(\delta, d, b)$, where $\delta\in \{1,2,\cdots,\delta_{\text{max}}\}$ is the current AoI, $d>0$ is the current distortion, and $b\in\{1,2,\cdots,b_{\text{max}}\}$ is the available energy at the sensor.
            We denote the set of all feasible states as the state space $\mathcal{S}$.
    \item \textbf{Actions:} An action defines the transmit power $P_k=p\in\{0,1,\cdots,b\}$ chosen by the sensor in the current block.
                When an action is taken, the state of the system will be changed in the next block.
            To be specific, the $\delta$ returns to $1$ if $p>0$ and goes to $\delta+1$ if $p=0$ while $d=D(p)$ if $p>0$ and keeps unchanged if $p=0$, in which $D(p)=\sigma^2_{\text{ob}} +  {(\sigma^2_\theta-\sigma^2_{\text{ob}})\sigma^2_{\text{ch}}}/
                                        {(\sigma^2_{\text{ch}}+p)}$ is given by \eqref{eq:theta_dist}.
    \item \textbf{Transition probabilities:} With an action $P_k=p$, the system transits from state $\boldsymbol{s}=(\delta, d, b)$ to state $\boldsymbol{t}=(\delta', d', b')$ with probability
        \begin{equation}\label{rt:P_st_0}
            \textsf{P}_{st}(p=0) = \left\{
                \begin{aligned}
                &\lambda, &&\text{if}~\boldsymbol{t}=(\delta+1,d,b+1), \\
                &1-\lambda        &&\text{if}~\boldsymbol{t}=(\delta+1,d,b),\\
                &0  &&\text{else},
                \end{aligned}
             \right.
        \end{equation}
        \begin{equation}\label{rt:P_st_1}
        \noindent    \textsf{P}_{st}(1\leq p\leq b) = \left\{
                \begin{aligned}
                &\lambda, &&\text{if}~\boldsymbol{t}=(1,D(p),b+1-p),\\
                &1-\lambda        &&\text{if}~\boldsymbol{t}=(1,D(p),b-p),\\
                &0  &&\text{else}.
                \end{aligned}
             \right.
        \end{equation}
        From \eqref{rt:P_st_0} and \eqref{rt:P_st_1}, it is inferred that any state $\boldsymbol{s} = (\delta, D(p), b)$ satisfying $\delta+(b_{\text{max}}-p)<b$ would always be inactive, since the age $\delta$ is too small for the sensor to accumulate enough energy for energy state $b$.
            Thus, we shall update the state space $\mathcal{S}$ by excluding these states.
    \item \textbf{Cost:} For a given state $\boldsymbol{S}_k=\boldsymbol{s}$ and action $P_k=p$, the cost $C(\boldsymbol{s}, p)$ is the weighted-sum AoI and distortion of the next block, i.e., $C(\boldsymbol{s}, p)=\delta_{k+1}+wd_{k+1}$.
            It is clear that $C(\boldsymbol{s}, p)=1+wD(p)$ if $p>0$ and $C(\boldsymbol{s}, p)=\delta+1+wd$ if $p=0$.
    \item \textbf{Policy:} A policy $\boldsymbol{\pi}$ is a rule for choosing actions (transmit power) for each state, i.e., a mapping from the state space $\mathcal{S}$ to the feasible power space $\{0,1,\cdots,b\}$.
\end{itemize}

For the online power control of the system, we shall seek such a policy $\boldsymbol{\pi}^*$ that minimizes the average cost of the system with any initial state $\boldsymbol{s}$, as shown in the following optimization problem.
\begin{equation}
    (\textbf{P}_4) ~~ \phi_{\boldsymbol{\pi}^*}(\boldsymbol{s}) = \min\limits_{\boldsymbol{\pi}} \frac1K\mathbb{E}
    \left[ \sum_{k=0}^\infty C(\boldsymbol{S}_k, P_k)\Big|\boldsymbol{S}_0=\boldsymbol{s}  \right]
\end{equation}
for all $\boldsymbol{s}\in\mathcal{S}$.

\subsection{Expected Total Discounted Cost}
As shown in \cite[Chap.~6.7,~\textit{Theorem}~6.17]{Ross-MDP-1970}, \textit{Problem} \textbf{P}$_4$ can be solved by the following functional equation.
\begin{equation}\label{eq:mdp_exp}
    g+h(\boldsymbol{s}) =\min\limits_{p}\left\{ C(\boldsymbol{s}, p) + \sum_{\boldsymbol{t}\in\boldsymbol{T}_{\boldsymbol{s},p}}
             \textsf{P}_{\boldsymbol{st}}(p) h(\boldsymbol{t}) \right\},
\end{equation}
where $g$ is a constant, $h(\boldsymbol{s})$ is a bounded function, and $\boldsymbol{T}_{\boldsymbol{s},p}$ is the set of possible states transited from state $\boldsymbol{s}$ when action $p$ is taken.

However, it is noted that equation \eqref{eq:mdp_exp} is not a contraction mapping.
    The searching process with \eqref{eq:mdp_exp}, therefore, may not be convergent or converge very slowly.
This motivates us to consider an alternative expected total $\alpha$-discounted cost,
\begin{equation}
    (\textbf{P}'_4) ~~ V_{\boldsymbol{\pi}_\alpha^*}(\boldsymbol{s}) = \min\limits_{\boldsymbol{\pi}} \frac1K\mathbb{E}
    \left[ \sum_{k=0}^\infty \alpha^k C(\boldsymbol{S}_k, P_k) \Big|\boldsymbol{S}_0=\boldsymbol{s} \right]
\end{equation}
for all $\boldsymbol{s}\in\mathcal{S}$, in which $\alpha<0<1$ is a discounting factor.
     Moreover, the $\alpha$-optimal policy $\boldsymbol{\pi}_\alpha^*$ and the $\alpha$-optimal cost function $V_\alpha(\boldsymbol{s})$ satisfies \cite[Chap.~6.7,~(24)]{Ross-MDP-1970},
\begin{equation}\label{eq:mdp_alpha}
    V_\alpha(\boldsymbol{s}) =\min\limits_{p}\left\{ C(\boldsymbol{s}, p) + \alpha \sum_{\boldsymbol{t}\in\boldsymbol{T}_{\boldsymbol{s},p}}
             \textsf{P}_{\boldsymbol{st}}(p) V_\alpha(\boldsymbol{t}) \right\}.
\end{equation}

Particularly, the following theorem shows that as $\alpha$ approaches unity, $\boldsymbol{\pi}_\alpha^*$ would converge to $\boldsymbol{\pi}^*$.
\begin{theorem} \label{th:mdp_alpha}
    For some sequence $\alpha_n\rightarrow1$, we have $h(\boldsymbol{s})=\lim_{n\rightarrow\infty} V_{\alpha_n}(\boldsymbol{s})-V_{\alpha_n}(\boldsymbol{s}_0)$, $g=\lim_{\alpha\rightarrow1}(1-\alpha) V_{\alpha}(\boldsymbol{s}_0)$, for any fixed reference state $\boldsymbol{s}_0$.
        In particular, \textit{Problem} \textbf{P}$_4$ and \textit{Problem} \textbf{P}$'_4$ share the same optimal policy.
\end{theorem}

\begin{proof}
    Since all of $\delta\in\{1,2,\cdots,\delta_{\text{max}}\}$, $b\in\{1,2,\cdots, b_{\text{max}}\}$, and $d=D(p)$ have finite number of elements, the state space $\mathcal{S}$ should be finitely large.
        Based on the transition probabilities given in \eqref{rt:P_st_0} and \eqref{rt:P_st_1}, it is seen that all the neighboring states (e.g., with $\delta'-\delta=1$, or $b'-b=1$) are connected with a strictly positive probability $\lambda$.
            It is also seen that each state $(\delta, D(p), b)$ can be connected with state $(\delta, D(p+1), b)$ through some intermediate states (e.g., several $(1, D(1), b-1)$ and a $(1, D(p+1), 0)$).
   Thus, the Markov chain is irreducible.
    According to \cite[Chap.~6.8,~\textit{Corollary}~6.20]{Ross-MDP-1970}, $V_{\alpha}(\boldsymbol{s})-V_{\alpha}(\boldsymbol{s}_0)$ would be uniformly bounded, and hence the conditions of \cite[Chap.~6.7,~\textit{Theorem}~6.17]{Ross-MDP-1970} are satisfied, which yield the results in \textit{Theorem} \ref{th:mdp_alpha} immediately.
\end{proof}

\subsection{Property of the $\alpha$-optimal Policy}
In this subsection, we investigate the property of the $\alpha$-optimal policy obtained through \eqref{eq:mdp_alpha}.

We define the expected future cost function as ${V}_\alpha(\boldsymbol{s},p) = \alpha \sum_{\boldsymbol{t} \in\boldsymbol{T}_{\boldsymbol{s},p}} \textsf{P}_{\boldsymbol{st}}(p) V_\alpha(\boldsymbol{t})$.
    Since the harvested energy is either zero or one unit in each block, the potential state set after the transition from state $\boldsymbol{s}=\{\delta, d,b\}$ only has two elements, i.e., $\boldsymbol{T}_{\boldsymbol{s},p}=\{\boldsymbol{t}_0,\boldsymbol{t}_1\}$.
Thus, $\widebar{V}_\alpha(\boldsymbol{s},p)$ can be rewritten as
\begin{equation}\label{df:future_exp_V}
    \widebar{V}_\alpha(\boldsymbol{s},p) = \alpha\big( \lambda V_\alpha(\boldsymbol{t}_1) +
                                        (1-\lambda)V_\alpha(\boldsymbol{t}_0)\big).
\end{equation}

First, we present the monotonicity of the $\alpha$-optimal cost function $V_\alpha(\boldsymbol{s})$ as follows.
\begin{proposition}\label{prop:4_vproperty}
    For each state $\boldsymbol{s}=(\delta, d, b)$, $V_\alpha(\boldsymbol{s})$ is
    \begin{itemize}
        \item non-decreasing with AoI $\delta$;
        \item non-decreasing with distortion $d$;
        \item non-increasing with energy state $b$;
        \item convex in energy state $b$.
    \end{itemize}
\end{proposition}

\begin{proof}
    See \textit{Appendix} \ref{prf:4_vproperty}.
\end{proof}

Next, we show the monotonicity of the optimal transmit power with respect to the energy state.

\begin{theorem}\label{th:2_pproperty}
    For each state $\boldsymbol{s}=(\delta,d,b)$, the optimal transmit power $p$ is non-decreasing with energy state $b$.
\end{theorem}
\begin{proof}
    See \textit{Appendix} \ref{prf:2_pproperty}.
\end{proof}

Moreover, the following theorem shows that the optimal transmit power has a threshold-structure with respect to AoI $\delta$ and distortion $d$.

\begin{theorem}\label{th:3_pproperty}
    Let $p$ be the optimal transmit power for state $\boldsymbol{s}=(\delta,d,b)$.
        For a state $\boldsymbol{s}'=(\delta',d',b)$ having the same energy status $b$ as $\boldsymbol{s}$, $p$ is also optimal if
    \begin{itemize}
      \item $p>0$, $\delta'>\delta$, and $d'=d$;
      \item $p=0$, $\delta'<\delta$, and $d'=d$;
      \item $p>0$, $\delta'=\delta$, and $d'>d$;
      \item $p=0$, $\delta'=\delta$, and $d'<d$.
    \end{itemize}
\end{theorem}

\begin{proof}
    See \textit{Appendix} \ref{prf:3_pproperty}.
\end{proof}

\textit{Theorem} \ref{th:3_pproperty} indicates that for a given energy state $b$, the optimal transmit power has only two possible values, i.e., $P_k=0$ if $\delta$ and $d$ are smaller than some thresholds or $P_k=p^*>0$ if $\delta$ and $d$ are larger than the thresholds.
    Therefore, the bi-valued and the threshold-type property of the optimal transmit power is very useful in searching the optimal policy $\pi^*$.
    In the sequel, however, we propose a matrix-calculation based algorithm to solve \textit{Problem} \textbf{P}$'_4$, as shown in \textit{Algorithm} \ref{alg:online_policy}.
That is, we present cost functions for all the states by a three-dimensional matrix and update it using matrix calculations, which is more efficient than updating the cost functions for the states one by one.
    In the algorithm, the MATLAB grammar is used, in which $\textbf{M}=\text{zeros}(m,n,l,j)$ returns an $m$-by-$n$-by-$l$-by-$j$ array of zeros and $\min(\textbf{M}, 4)$ is the minimization operation along the fourth dimension.

In particular, \textit{Theorem} \ref{th:mdp_alpha} guarantees that the results shown in \textit{Proposition} \ref{prop:4_vproperty}, \textit{Theorem} \ref{th:2_pproperty}, and \textit{Theorem} \ref{th:3_pproperty} also hold true for the expected total discounted cost problem, i.e., \textit{Problem} \textbf{P}$_4$.

\begin{algorithm}[!t]
\algsetup{linenosize=\scriptsize}
\scriptsize
\caption{Online policy searching}
\begin{algorithmic}[1]\label{alg:online_policy}
\REQUIRE ~~\\
    \STATE Set $\Delta v=+\infty$, $\varepsilon=10^{-3}$;
    \STATE Initialize cost function matrix $\textbf{V}_\alpha=\text{zeros}(\delta_{\text{max}}, d_{\text{max}}, b_{\text{max}}+1) $;
    \STATE Initialize power matrix $\textbf{P}_\alpha=\text{zeros}(\delta_{\text{max}}, d_{\text{max}}, b_{\text{max}}+1) $;
\ENSURE ~~\\
\WHILE{$\Delta v>\varepsilon$}
    \STATE $\textbf{V}_\alpha^{\text{old}} = \textbf{V}_\alpha$, $\textbf{V}_\alpha^{p} =\text{zeros}(\delta_{\text{max}}, d_{\text{max}}, b_{\text{max}}+1) $;
     \STATE $\textbf{V}_\alpha^{\forall p}=\text{zeros}(\delta_{\text{max}}, d_{\text{max}}, b_{\text{max}}+1, b_{\text{max}}+1) $;
    \FOR {$p=0$ to $b_{\text{max}}$}
        \STATE $\textbf{V}_\alpha^{p}= C(\boldsymbol{s}, p) + \alpha \sum_{\boldsymbol{t}\in\boldsymbol{T}_{\boldsymbol{s},p}}
             \textsf{P}_{\boldsymbol{st}}(p) V_\alpha(\boldsymbol{t})$
             and set element ${V}_\alpha(\boldsymbol{s})=+\infty$ if $p>b$ is true for state $\boldsymbol{s}$;
        \STATE $\textbf{V}_\alpha^{\forall p}(:,:,:,p+1) = \textbf{V}_\alpha^p$;
    \ENDFOR
    \STATE $[\textbf{V}_\alpha, \textbf{P}_\alpha] = \min\big(\textbf{V}_\alpha^{\forall p},4\big)$;
    \STATE $\Delta v=\max\big(\max(\max|\textbf{V}_\alpha-\textbf{V}_\alpha^{\text{old}}|)\big)$
\ENDWHILE
 \STATE \textbf{Output:}  $\textbf{V}_\alpha, \textbf{P}_\alpha$.
\end{algorithmic}
\end{algorithm}

%
%
%

\begin{figure}[!t]
\centering
\includegraphics[width=3.7in]{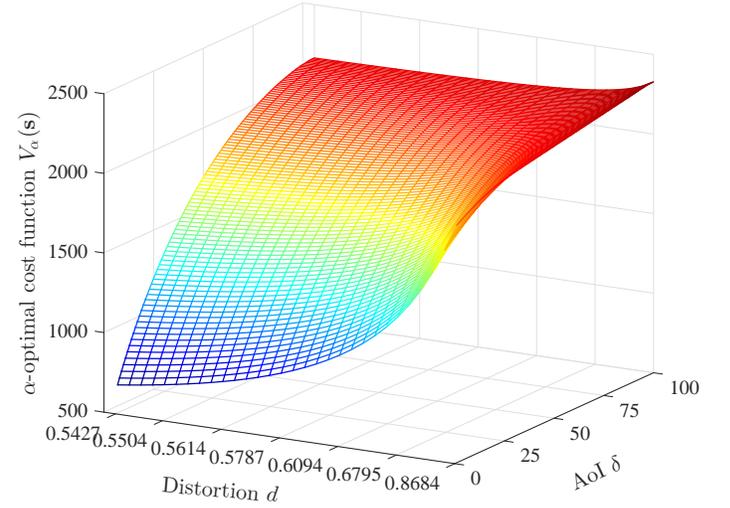}
\caption{$\alpha$-optimal cost function $V_\alpha(\boldsymbol{s})$ versus AoI and distortion ($b=0$).  } \label{fig:v_a_d_1}
\end{figure}

\begin{figure}[!t]
\centering
\includegraphics[width=3.7in]{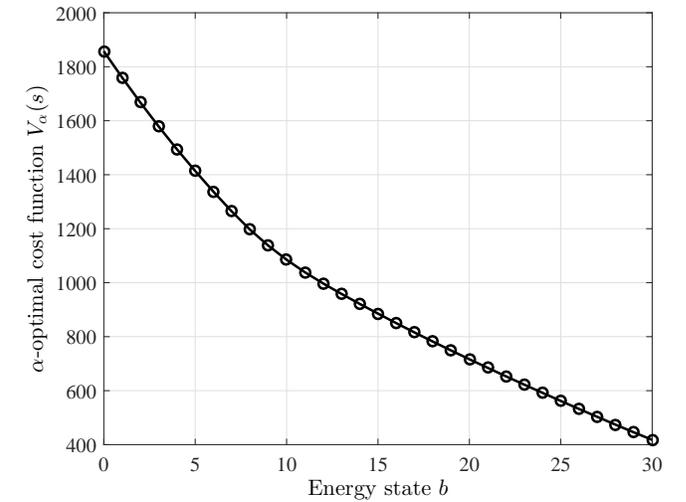}
\caption{$\alpha$-optimal cost function $V_\alpha(\boldsymbol{s})$ versus energy state $b$, in which $(\delta,d)=(40,0.6094)$.  } \label{fig:v_b}
\end{figure}

\begin{figure}[!t]
\centering
\includegraphics[width=3.7in]{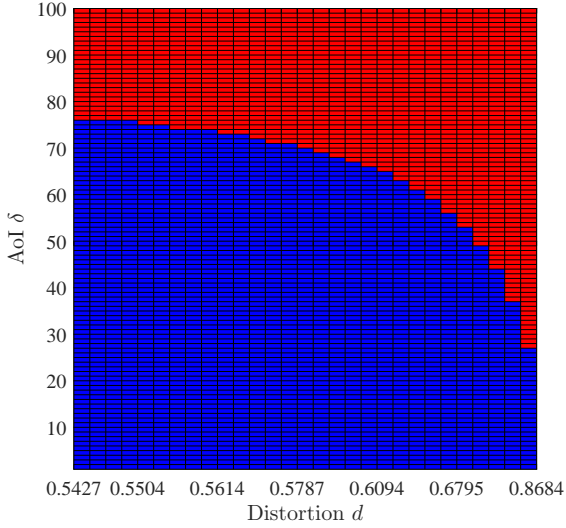}
\caption{Optimal transmit power versus AoI and distortion ($b=3$).  } \label{fig:p_a_d_1}
\end{figure}

\section{Numerical Results}\label{sec:6_simulation}
In this section, we present the developed results above through numerical and Monte Carlo simulations.
    Without loss of generality, we set the source signal power to $\sigma^2_\theta=1$, the observation noise power to $\sigma^2_{\text{ob}} = 0.5$, the channel noise power to $\sigma^2_{\text{ch}} = 2.8$.
Under this setting, we have $w_0=12.8929$ and $\sigma^2_{\text{ob}0} = 0.6777$.
    For the online power control, we set the discount factor to $\alpha=0.999$.

\subsection{Performance of Online and Offline Power Control}
First, we set $w=200$, $\delta_{\text{max}}=100$, $b_{\text{max}}=30$, and run \textit{Algorithm} \ref{alg:online_policy} to investigate the monotonicity of the $\alpha$-optimal cost function $V_\alpha(\boldsymbol{s})$ and the threshold property of the optimal transmit power.
    In Fig. \ref{fig:v_a_d_1}, the energy state is set as $b=0$ and it is observed that $V_\alpha(\boldsymbol{s})$ is increasing both with AoI and distortion (cf. \textit{Proposition} \ref{prop:4_vproperty}.1 and 3.2).
It should be noted that the monotonicity of $V_\alpha(\boldsymbol{s})$ also hold for any $b>0$.
It is further observed in Fig. \ref{fig:v_b} that $V_\alpha(\boldsymbol{s})$ is convex and decreasing with respect to energy state $b$, (cf. \textit{Proposition} \ref{prop:4_vproperty}.3 and 3. 4).
    Fig. \ref{fig:p_a_d_1} presents the optimal transmit power for each states, in which the energy state is set to $b=3$.
It is seen that the optimal transmit power is $P=3$ only if the AoI $\delta$ and the distortion $d$ are large, i.e.,
     the optimal transmit power have a threshold-type property with respect to $\delta$ and $d$ (cf. \textit{Theorem} \ref{th:3_pproperty}).
In particular, the optimal transmit power is dominated by $\delta$, since $\delta$ is much larger than $d$.
    Moreover, for some given energy states, we show the boundaries for transmit power to be positive, as shown in Fig. \ref{fig:p_a_d_2}.
It is observed that when $b$ is increased, we have a lower boundary, which means that the sensor is more likely to perform a block of observation and transmission.
    Furthermore, for the given AoI-distortion pair $(\delta,d)=(40,0.6094)$, the left figure in Fig. \ref{fig:p_a_d_2} shows that the optimal transmit power is non-decreasing with energy state $b$.

\begin{figure}[!t]
\centering
\includegraphics[width=3.7in]{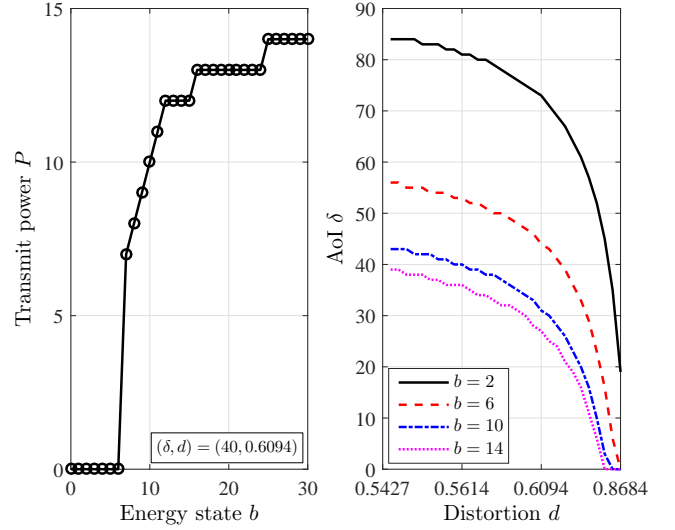}
\caption{Optimal transmit power and boundaries for it to be positive.  } \label{fig:p_a_d_2}
\end{figure}

In Fig. \ref{fig:age_distortion_trade_4_cases}, we plot the trade-off between the average AoI and the average distortion for each weighting coefficient $w\in\{w_0:25:500\}$ and for all the four schemes under test.
    For the online policy, we consider a period of $K=10^5$ blocks and solve the optimal power control  for each $w$ using \textit{Algorithm} \ref{alg:online_policy}.
    It is observed that the performance (the dotted curve) of the online power control closely approaches that (the solid curve) of the save-and-transmit policy, i.e., the performance limit of the system.
In the offline power control, we consider a period of $K=100$ blocks due to the computational complexity of \textit{Algorithm} \ref{alg:genetic}.
    For each $w$, the genetic algorithm is run for $N_{\text{iter}}=200$ iterations, in which $N_{\text{pop}}=100$ chromosomes are considered.
We set the selection rate to $q_{\text{sel}}=0.5$ and choose a half of the chromosomes as potential parent chromosomes.
    The crossover depth is $D_{\text{cross}}=34$, i.e., each chromosome would be randomly selected and changed with $34$ randomly chosen gene-instances in each iteration (to generate a child chromosome).
To evaluate the fitness of a chromosomes using \eqref{df:fitness}, we first identify the block allocation $\{X_l\}$ from the chromosome and then solve the corresponding optimal power allocation $\{P_l\}$ using \textit{Algorithm} \ref{alg:wb_wf}.
    It is observed that the offline power control (desh-dotted curve) does not perform as well as other schemes.
The reason is that the considered period is a bit too short for the scheme to approach the performance limit.
    Also, the genetic algorithm is not guaranteed to find the optimal solution in a finite number (e.g., $N_{\text{iter}}=200$) of iterations with finite number of chromosomes (e.g., $N_{\text{pop}}=100$).

\begin{figure}[!t]
\centering
\includegraphics[width=3.7in]{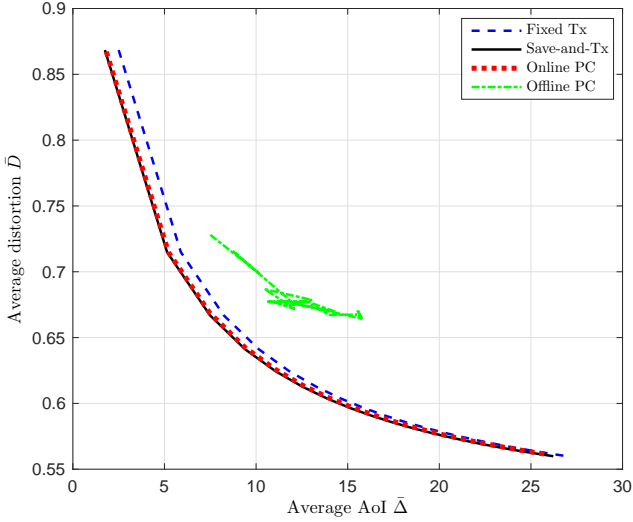}
\caption{AoI-distortion trade-off.  } \label{fig:age_distortion_trade_4_cases}
\end{figure}

\subsection{Timeliness and Distortion in Fading Sensing Systems}\label{subsec:6_B}
In this subsection, we consider a sensing system with block Rayleigh fading.
    That is, the power gain $\rho$ of the channel does not change within each block and varies independently among blocks with probability density function $f_\rho(x)=\exp(-x/\sigma^2_{\text{fd}})/\sigma^2_{\text{fd}}$, in which we have $x>0$ and $0<\sigma^2_{\text{fd}}<1$.
Note that the randomness in the channel gain does not change the expression of average AoI.
    For the fixed power transmission and the save-and-transmit policy, therefore, we have
\begin{align}
    \widebar{\Delta}_{\text{fx}} = \frac {P_{\text{fx}}+1}{2\lambda}~\text{and}~
    {\Delta}_{\text{sv}} = \frac {P_{\text{sv}}+\lambda}{2\lambda}.
\end{align}

Moreover, the expected distortion would be
\begin{align}
    &\mathbb{E}[D_{\text{fx}}^{\text{fading}}]=\mathbb{E}[D_{\text{sv}}^{\text{fading}}]=
    \sigma^2_{\text{ob}} +  \mathbb{E}\left[\frac{(\sigma^2_\theta-\sigma^2_{\text{ob}})\sigma^2_{\text{ch}}}
                                        {\sigma^2_{\text{ch}}+\rho P}\right]\\
    &=\sigma^2_{\text{ob}} + (\sigma^2_\theta-\sigma^2_{\text{ob}}) z e^z \text{E}_1(z),
\end{align}
where $z=\frac{\sigma^2_{\text{ch}}}{\sigma^2_{\text{fd}} P}$ and $\text{E}_1(z)=\int_z^\infty e^{-u}/u du$ is the first order exponential integral.

Therefore, \textit{Problem} \textbf{P}$_0$ turns to be
\begin{align}
\label{prob:5}
  (\textbf{P}_5)~~~
    \min\limits_{P}~~~&\frac{P+1}{2\lambda} + w\sigma^2_{\text{ob}} +w(\sigma^2_\theta-\sigma^2_{\text{ob}}) z
            e^z \text{E}_1(z)  \\
 \label{prob:1_1}
        \text{subject~to}~~~ & P\geq 1.
\end{align}

By taking the derivative of the objective function with respect to $P$, we have
\begin{equation}
    \frac{\partial \mathcal{J}}{\partial P} = \frac{1}{2\lambda} -
            \frac{w(\sigma^2_\theta-\sigma^2_{\text{ob}})}{ P}
             \Big((z^2+z)e^z\text{E}_1(z) - z \Big).
\end{equation}

By setting the derivative to zero, we  have
\begin{equation}\label{dr:fading_P}
    P = 2\lambda w(\sigma^2_\theta-\sigma^2_{\text{ob}}) \Big((z^2+z)e^z \text{E}_1(z) - z \Big).
\end{equation}

Thus, we can solve the optimal transmit power by applying \eqref{dr:fading_P} to any initial power (e.g., $P^{(0)}=1$ and $z^{(0)}={\sigma^2_{\text{ch}}}/{\sigma^2_{\text{fd}}}$) iteratively, which does not stop until the process converges.

\begin{figure}[!t]
\centering
\includegraphics[width=3.7in]{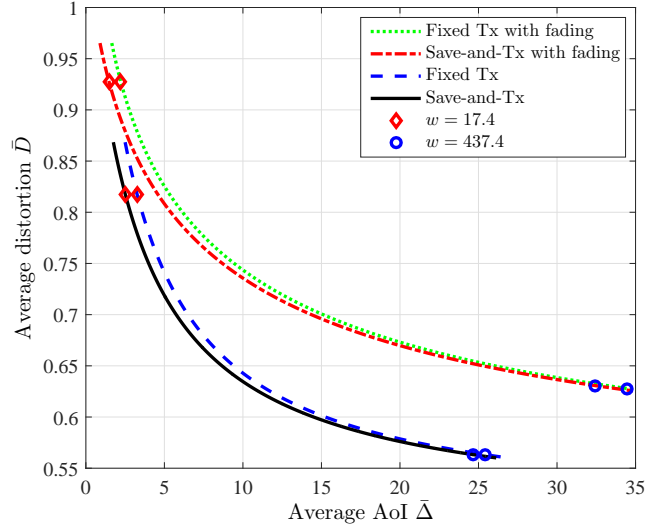}
\caption{AoI-distortion trade-off with block Rayleigh fading.  } \label{fig:age_distortion_trade_fading}
\end{figure}

We set the expected fading channel power gain to $\mathbb{E}[\rho]=\sigma^2_{\text{fd}}=0.7$ and plots the AoI-distortion trade-off of the fading sensing system in Fig. \ref{fig:age_distortion_trade_fading}.
    When a small weighting coefficient (e.g., $w=17.4$) is used, it is observed that the average AoI of the fading system is smaller than that of the non-fading system.
The reason is that compared with the non-fading system, using a larger transmit power yields less reduction in average distortion in fading systems due to the randomness of channel gains.
    When $w$ is relatively small, the optimizer would be more concentrated on the average AoI of the system.
On the contrary, we have to accumulate more energy to reduce the average distortion if $w$ is relatively large (e.g., $w=437.4$).
    However, the fading system is not as efficient as the non-fading system, and thus have a larger average AoI and a larger average distortion.
As is expected, Fig. \ref{fig:J_fading} shows that the fading system is inferior to the non-fading system in terms of achievable weighted-sum AoI and distortion, especially when $w$ is relatively large.

\begin{figure}[!t]
\centering
\includegraphics[width=3.7in]{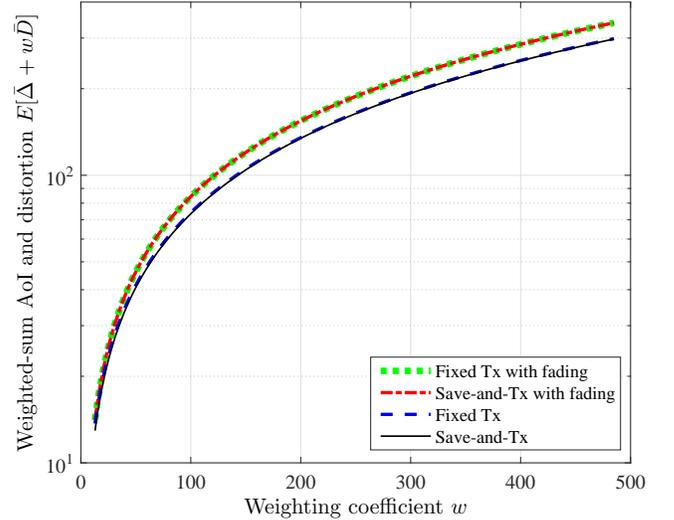}
\caption{Weighted-sum AoI and distortion with block Rayleigh fading ($\mu=0.7$).  } \label{fig:J_fading}
\end{figure}

\section{Conclusion}\label{sec:5_conclusion}
In this paper, we have investigated the timeliness-distortion trade-off of an energy harvesting powered sensing system.
    From the point view of  Shannon's information theory, seeking the limits of communication efficiency/reliability and designing limit-approaching schemes have been the core field of research for the communication society.
In recent years, however, the role of communications is changing from a relatively independent research to an indispensable support for human sensing, which originates from human beings' curiosity on unknowns.
    We sense the world by seeing and feeling the ambient environments, by accessing texts, pictures, videos from books and internet, and also passively by the suggestions from friends and recommending systems, which may collaboratively be referred to as the ubiquitous sensing.
During the sensing process, we may no longer need to deep into the performance of communications and computations.
    What interests us would be the timeliness, the accuracy, and the credibility of the sensing.
In this paper, therefore, we have focused on the timeliness and the distortion of IoT systems.
    By minimizing the average weighted-sum AoI and distortion, we presented optimal solutions for several observation and transmission schemes.
As the future work, we are interested in evaluating the performance limits and designing optimal schemes for systems with more intelligent sensing, e.g., a system including active sensing from deliberately deployed sensors and information search engines, as well as passive sensing from recommending systems.

\appendix

\renewcommand{\theequation}{\thesection.\arabic{equation}}
\newcounter{mytempthcnt}
\setcounter{mytempthcnt}{\value{theorem}}
\setcounter{theorem}{2}
\addcontentsline{toc}{section}{Appendices}\markboth{APPENDICES}{}

\subsection{Proof of \textit{Proposition} \ref{prop:fixed}}  \label{prf:fixed}
\begin{proof}
Since the sensor harvests one unit of energy with probability $\lambda$ in each block, i.e., following the Bernoulli process, it is clear that $\tau_{\text{H}k}$ follows a negative binomial distribution with parameters $P_{\text{fx}}$ and $\lambda$.
    In particular, we have $\Pr\{ \tau_{\text{H}l}=j \}= { {P_{\text{fx}}-1}\choose{j-1} } \lambda^{P_{\text{fx}}} (1-\lambda)^{j-P_{\text{fx}}}$ for $j=P_{\text{fx}},P_{\text{fx}}+1,\cdots$.
Moreover, the first and the second order moments of $\tau_{\text{H}l}$ are given by
\begin{align}\label{apx:a_1}
    \mathbb{E}[\tau_{\text{H}l}] = \frac{P_{\text{fx}}}{\lambda} ~\text{and}~
    \mathbb{E}[\tau_{\text{H}l}^2] = \frac{P_{\text{fx}}}{\lambda^2}(P_{\text{fx}}+1-\lambda).
\end{align}

As $K$ goes to infinity, we see from Fig. \ref{fig:aoi} that the average AoI can be calculated by
\begin{align}
    \widebar\Delta_{\text{fx}} = \frac1K \sum_{k=1}^K \Delta_k = \frac LK  \frac1L\sum_{l=1}^L Q_l
        = \frac{1}{\mathbb{E}[Y_l]} \mathbb{E}[Q_l],
\end{align}
where
\begin{align}
    Q_l &=\frac12 Y_l(Y_l+1)] = \frac{Y_l}2+
                \frac{Y_l^2}2
\end{align}
is the  area under the AoI curve during inter-departure time $Y_k$.
    It is also observed from Fig. \ref{fig:aoi} that for each busy block, we have $X_k= Y_k=\tau_{\text{H}k}$.
Thus, the average AoI would be
\begin{equation}\label{apx:delta_fix}
    \widebar\Delta_{\text{fx}} =\frac12+
                \frac{\mathbb{E}[\tau_{\text{H}k}^2]}{2\mathbb{E}[\tau_{\text{H}k}]} =  \frac{P_{\text{fx}}+1}{2\lambda}.
\end{equation}

By combing \eqref{eq:theta_dist} and\eqref{eq:delta_fix}, we have
\begin{align}
\mathcal{J} &= \frac1K \sum_{k=1}^{K} (\Delta_k+wD_k) \\
                      \label{eq:dis_fix_gammaov_small}
                      & = \frac{P_{\text{fx}}+1}{2\lambda} + w\sigma^2_{\text{ob}} +  \frac{w(\sigma^2_\theta-\sigma^2_{\text{ob}})\sigma^2_{\text{ch}}}
                                        {\sigma^2_{\text{ch}}+P_{\text{fx}}}.
\end{align}

The derivative of $\mathcal{J} $ with respect to $P_{\text{fx}}$ is given by
\begin{equation}\label{apx:J_aoi_fx}
    \frac{\partial \mathcal{J}}{\partial P_{\text{fx}}} = \frac{1}{2\lambda} - \frac{w(\sigma^2_\theta-\sigma^2_{\text{ob}})\sigma^2_{\text{ch}}}
                                        {(\sigma^2_{\text{ch}}+P_{\text{fx}})^2}.
\end{equation}

First, we observe that  $\frac{\partial \mathcal{J}}{\partial P_{\text{fx}}}$ is positive for any $P_{\text{fx}}\geq1$ and any $\sigma^2_{\text{ob}}$ if the condition $w\leq w_0=(1+\sigma^2_{\text{ch}})^2 /(2\lambda\sigma^2_\theta\sigma^2_{\text{ch}})$ is satisfied.
    In the case $w> w_0$, $\frac{\partial \mathcal{J}}{\partial P_{\text{fx}}}$ also is positive for any $P_{\text{fx}}\geq1$ if $\sigma^2_{\text{ob}}\geq \sigma^2_{\text{ob}0} =\sigma^2_\theta(1-w_0/w) $.
For these cases, therefore, the optimal transmit power would be $P_{\text{fx}}=1$.

Second, if $w> w_0$ and $\sigma^2_{\text{ob}}< \sigma^2_{\text{ob}0}$ are satisfied, the optimal transmit power should be the solution to $\frac{\partial \mathcal{J}}{\partial P_{\text{fx}}}=0$, which leads to
    \begin{equation}
        P_{\text{fx}} = \sqrt{2\lambda w(\sigma^2_\theta-\sigma^2_{\text{ob}})\sigma^2_{\text{ch}}} - \sigma^2_{\text{ch}}.
    \end{equation}
    This completes the proof of \textit{Proposition} \ref{prop:fixed}.
\end{proof}

\subsection{Proof of \textit{Proposition} \ref{prop:3_off p}}  \label{prf:3_off p}
\begin{proof}
    It is clear that \textit{Problem} \textbf{P}$'_3$ is convex in each $P_l$.
The corresponding Lagrangian can be expressed as
\begin{align}
    \mathcal{L} = \frac1{K} \sum_{l=1}^L  \frac{Y_{l}}
                                        {\sigma^2_{\text{ch}}+P_{l-1}}
                        + \sum _{l=1}^L \mu_l \left(\sum_{i=1}^l P_i - \sum_{i=1}^{l}\sum_{j=1}^{X_{l}} E_{ij} \right)
\end{align}
Taking the derivative with respect to $P_l$ and set it to zero, we have
\begin{equation} \label{apx:pl_deriv}
    \frac1{K}  \frac{-Y_{l+1}}
                                        {(\sigma^2_{\text{ch}}+P_{l})^2} + \nu_l =0,
\end{equation}
where $\nu_l=\sum _{j=l}^L \mu_j$ is the water-level for the $l$-th busy block.
Note that $-\nu_l$ is equal to the first order derivative of the objective function in \textit{Problem} \textbf{P}$'_3$ .
    Thus, the objective function would be minimized when the water-levels are as close to each other as possible under constraint \eqref{prob:3_2_1}, i.e., the marginal gain of increasing each $P_l$ is almost the same.

By solving $P_l$ from \eqref{apx:pl_deriv}, the proof of \textit{Proposition} \ref{prop:3_off p} would be completed.

\end{proof}

\subsection{Proof of \textit{Proposition} \ref{prop:4_vproperty}}  \label{prf:4_vproperty}
\begin{proof}
\subsubsection{$V_\alpha(\boldsymbol{s})$  is non-decreasing with $b$}
    Note that $V_\alpha(\boldsymbol{s})$ is a mapping from state space $\mathcal{S}$ to the real space.
        We then define a functional $T_f$ in the following manner \cite[Chap.~6.2,~(11)]{Ross-MDP-1970}.
        \begin{equation}\label{df:tf}
            (T_\alpha u)(\boldsymbol{s}) = \min\limits_{p} \left\{C(\boldsymbol{s}, p) + \alpha \sum_{\boldsymbol{t} \in\boldsymbol{T}_{\boldsymbol{s},p}} \textsf{P}_{\boldsymbol{st}}(p) u(\boldsymbol{t}) \right\}.
        \end{equation}
    That is, for a bounded real-valued function $u$, $T_\alpha u$ is the function whose value at state $\boldsymbol{s}$ is given by \eqref{df:tf}.
        As shown in \cite[Chap.~6.2,~\textit{Theorem}~6.5]{Ross-MDP-1970}, $T_\alpha u$ is a contraction mapping.
    According to \cite[Chap.~6.2,~\textit{Lemma}~6.2]{Ross-MDP-1970}, therefore, we have $\lim_{n\rightarrow\infty} T_\alpha^n u\rightarrow V_\alpha$.
        That is, the function $V_\alpha$ can be obtained by successively applying $T_\alpha$ to any initial bounded real-valued function $u$.
    Thus, we shall start from the function $u(\boldsymbol{s})=0$ and prove the result using mathematical induction.

    First, it is observed that $u(\boldsymbol{s})=0$ is non-decreasing with $\delta$ for any state $\boldsymbol{s}\in\mathcal{S}$.

    Second, after applying $T_\alpha$ to $u$ once, we have
    \begin{equation}
        T_\alpha^1 u(\boldsymbol{s}) = \min\limits_{p} C(\boldsymbol{s}, p).
    \end{equation}
    For two states $\boldsymbol{s}_1={(\delta,d,b_1)}$ and $\boldsymbol{s}_2={(\delta,d,b_2)}$ in which $b_1<b_2$, it is clear that
    \begin{align}
        &T_\alpha^1 u(\boldsymbol{s}_1)= \min\limits_{p\in\{0,1,\cdots,b_1\}} C(\boldsymbol{s}_1, p) \\
            & = \min\{\delta+1+wd, 1+wD(b_1)\}\\
            & \geq \min\{\delta+1+wd, 1+wD(b_2)\} \\
            & = \min\limits_{p\in\{0,1,\cdots,b_2\}} C(\boldsymbol{s}_2, p)\\
            &= T_\alpha^1 u(\boldsymbol{s}_2).
    \end{align}
    That is, $T_\alpha^1 u(\boldsymbol{s})$ is non-increasing with $b$.

    Third, we assume that $T_\alpha^n u(\boldsymbol{s})$ is non-increasing with $b$, i.e., $T_\alpha^n u(\boldsymbol{s}_1) \geq T_\alpha^n u(\boldsymbol{s}_2) $.
        By applying $T_\alpha$ to $u$ once more, it is clear that
    \begin{align}
        & \min\limits_{p\in\{0,1,\cdots,b_1\}} \left\{ C(\boldsymbol{s}_1, p) +
                        \alpha \sum_{\boldsymbol{t} \in\boldsymbol{T}_{\boldsymbol{s}_1,p}} \textsf{P}_{\boldsymbol{st}}(p) T_\alpha^n u(\boldsymbol{t}) \right\}\\
\geq&  \min\limits_{p\in\{0,1,\cdots,b_2\}} \left\{ C(\boldsymbol{s}_2, p) +
                        \alpha \sum_{\boldsymbol{t} \in\boldsymbol{T}_{\boldsymbol{s}_2,p}} \textsf{P}_{\boldsymbol{st}}(p) T_\alpha^n u(\boldsymbol{t}) \right\}.
    \end{align}
    That is, $T_\alpha^{n+1} u(\boldsymbol{s}_1) \geq T_\alpha^{n+1} u(\boldsymbol{s}_2)$ and $T_\alpha^{n+1} u(\boldsymbol{s})$ is non-increasing with $b$.

    Therefore, we see that $V_\alpha(\boldsymbol{s})=\lim_{n\rightarrow\infty} T_\alpha^{n} u(\boldsymbol{s}) $ is non-increasing with $b$.

    Likewise, it can be readily proved that $V_\alpha(\boldsymbol{s})$ is non-decreasing with $\delta$ and $d$.

    \subsubsection{$V_\alpha(\boldsymbol{s})$  is convex with $b$}
    First, by applying $T_\alpha$ to $u=0$, we have
    \begin{align}
        &T_\alpha^1 u(\boldsymbol{s}) = \min\limits_p C(\boldsymbol{s}, p) \\
        &~= \min\{\delta+1+wd, 1+wD(b_1) \},
    \end{align}
    which is convex in $b$.

    Second, we assume that $T_\alpha^n u(\boldsymbol{s})$ is convex with $b$.
        Note that $T_\alpha^{n+1} u(\boldsymbol{s})$ can be expressed as
    \begin{equation}
         \min\limits_{p\in\{0,1,\cdots,b\}} \left\{ C(\boldsymbol{s}, p) +
                        \alpha \sum_{\boldsymbol{t} \in\boldsymbol{T}_{\boldsymbol{s},p}} \textsf{P}_{\boldsymbol{st}}(p) T_\alpha^n u(\boldsymbol{t}) \right\}.
    \end{equation}
    Note also that $C(\boldsymbol{s}, p)$ is convex in $b$ for each given $p$, $T_\alpha^n u(\boldsymbol{s})$ is convex with $b$ as assumed, and the minimizing operation is convexity preserving.
        Thus, $T_\alpha^{n+1} u(\boldsymbol{s})$ is also convex with $b$.

    Finally, we have $V_\alpha(\boldsymbol{s})$ is convex in $b$ since $V_\alpha(\boldsymbol{s})=\lim_{n\rightarrow\infty} T_\alpha^{n} u(\boldsymbol{s}) $.
        This completes the proof of \textit{Proposition} \ref{prop:4_vproperty}.
\end{proof}

\subsection{Proof of \textit{Theorem} \ref{th:2_pproperty}} \label{prf:2_pproperty}
\begin{proof}
    Let $\boldsymbol{s}_1={(\delta,d,b_1)}$ and $\boldsymbol{s}_2={(\delta,d,b_2)}$ be two states in which $b_1<b_2$.
        Denote the $\alpha$-optimal transmit power for the two states as $p_1$ and $p_2$, respectively.
    To prove the theorem, we need show that $p_1\leq p_2$.

    We assume that $p_1> p_2$ and shall prove the result via contradiction.
        Since we have $b_1<b_2$, $p_1$ would also be a feasible transmit power for state $\boldsymbol{s}_2$.
         It has been shown in \textit{Proposition} \ref{prop:4_vproperty} that $V_\alpha(\boldsymbol{s})$ is convex and non-increasing with energy state $b$.
    Since $V_\alpha(\boldsymbol{s})$ is strictly positive for each state, $V_\alpha(\boldsymbol{s})$ must be decreasing more and more slowly as $b$ is increased, i.e., $V_\alpha(\delta, d, b_1-p)-V_\alpha(\delta, d, b_1)\geq V_\alpha(\delta, d, b_2-p)-V_\alpha(\delta, d, b_2)$ hold true for each transmit power $p\leq\min(b_1,b_2)$ (which is the reduction in energy).
        Since $\widebar{V}_\alpha(\boldsymbol{s},p)$ (cf. \eqref{df:future_exp_V}) is linear combination of $V_\alpha(\boldsymbol{s})$, we also have
    \begin{equation}\label{apx:contra_1}
        \widebar{V}_\alpha(\boldsymbol{s}_1, p_1) - \widebar{V}_\alpha(\boldsymbol{s}_1, p_2) \geq \widebar{V}_\alpha(\boldsymbol{s}_2, p_1) - \widebar{V}_\alpha(\boldsymbol{s}_2, p_2).
    \end{equation}

    On the other hand, since $p_1$ is optimal for $\boldsymbol{s}_1$, we have
    \begin{align}
         C(\boldsymbol{s}_1, p_1) + \widebar{V}_\alpha(\boldsymbol{s}_1, p_1)
                \leq C(\boldsymbol{s}_1, p_2) + \widebar{V}_\alpha(\boldsymbol{s}_1, p_2),
    \end{align}
    which is equivalent to
    \begin{equation} \label{apx:contra_21}
        C(\boldsymbol{s}_1, p_2) - C(\boldsymbol{s}_1, p_1) \geq \widebar{V}_\alpha(\boldsymbol{s}_1, p_1) -\widebar{V}_\alpha(\boldsymbol{s}_1, p_2)
    \end{equation}

    Likewise, we have
    \begin{equation}\label{apx:contra_22}
        C(\boldsymbol{s}_2, p_2) - C(\boldsymbol{s}_2, p_1) \leq \widebar{V}_\alpha(\boldsymbol{s}_2, p_1) -\widebar{V}_\alpha(\boldsymbol{s}_2, p_2)
    \end{equation}
    since $p_2$ is optimal for $\boldsymbol{s}_2$.

    Moreover, we note that for each $p$, we have $C(\boldsymbol{s}_1, p)=C(\boldsymbol{s}_2, p)$ since $\boldsymbol{s}_1$ is different from $\boldsymbol{s}_2$ in $b$ while $C(\boldsymbol{s}, p)$ is independent from $b$.
        Thus, we have
    \begin{equation}\label{apx:contra_23}
        C(\boldsymbol{s}_1, p_2) - C(\boldsymbol{s}_1, p_1) = C(\boldsymbol{s}_2, p_2) - C(\boldsymbol{s}_2, p_1).
    \end{equation}

    By combing \eqref{apx:contra_21}, \eqref{apx:contra_22}, and \eqref{apx:contra_23}, we have
    \begin{equation}\label{apx:contra_24}
        \widebar{V}_\alpha(\boldsymbol{s}_1, p_1) -\widebar{V}_\alpha(\boldsymbol{s}_1, p_2)
        \leq \widebar{V}_\alpha(\boldsymbol{s}_2, p_1) -\widebar{V}_\alpha(\boldsymbol{s}_2, p_2),
    \end{equation}
    which is contradict with \eqref{apx:contra_1}.

    Therefore, the assumption $p_1>p_2$ cannot be true and we have $p_1\leq p_2$, which completes the proof of the theorem.
\end{proof}

\subsection{Proof of \textit{Theorem} \ref{th:3_pproperty}} \label{prf:3_pproperty}
\begin{proof}
\subsubsection{The case of $p>0$, $\delta'>\delta$, and $d'=d$}
Since $p$ is optimal for $\boldsymbol{s}$,  for any $q \neq p$, we have
\begin{align}
                    C(\boldsymbol{s}, p) + \widebar{V}_\alpha(\boldsymbol{s}, p)
     \leq C(\boldsymbol{s}, q) + \widebar{V}_\alpha(\boldsymbol{s}, q). 
\end{align}

For any transmit power $q>0$, the AoI will return to one after the transition, regardless the current AoI.
    Thus, we have $C(\boldsymbol{s}, q)=1+wD(q)=C(\boldsymbol{s}', q)$ and $\widebar{V}_\alpha(\boldsymbol{s}, q)=\widebar{V}_\alpha(\boldsymbol{s}', q)$, which indicates
\begin{align}\label{apx:delta_thrd_1}
{V}_\alpha(\boldsymbol{s}')|_ p &=C(\boldsymbol{s}', p) + \widebar{V}_\alpha(\boldsymbol{s}', p)\\
    &=C(\boldsymbol{s}, p) + \widebar{V}_\alpha(\boldsymbol{s}, p) \\
    & \leq C(\boldsymbol{s}, q) + \widebar{V}_\alpha(\boldsymbol{s}, q) \\
    \label{apx:delta_thrd_11}
    & = C(\boldsymbol{s}', q) + \widebar{V}_\alpha(\boldsymbol{s}', q) = {V}_\alpha(\boldsymbol{s}')|_q.
\end{align}

For the case $q=0$, we have $C(\boldsymbol{s}, 0)=\delta+1+wd\leq \delta'+1+wd=C(\boldsymbol{s}', q)$.
    We also have $\widebar{V}_\alpha(\boldsymbol{s}, 0)\leq\widebar{V}_\alpha(\boldsymbol{s}', 0)$ since $\widebar{V}_\alpha(\boldsymbol{s}, p)$ is a linear combination of functions ${V}_\alpha(\boldsymbol{s})$ while  ${V}_\alpha(\boldsymbol{s})$ is non-decreasing with $\delta$.
Hence, we have,
\begin{align}\label{apx:delta_thrd_2}
{V}_\alpha(\boldsymbol{s}')|_ p &=C(\boldsymbol{s}', p) + \widebar{V}_\alpha(\boldsymbol{s}', p)\\
    &=C(\boldsymbol{s}, p) + \widebar{V}_\alpha(\boldsymbol{s}, p) \\
    & \leq C(\boldsymbol{s}, 0) + \widebar{V}_\alpha(\boldsymbol{s}, 0) \\
    \label{apx:delta_thrd_22}
    & \leq C(\boldsymbol{s}', 0) + \widebar{V}_\alpha(\boldsymbol{s}', 0) = {V}_\alpha(\boldsymbol{s})|_0.
\end{align}

By combining \eqref{apx:delta_thrd_1}, \eqref{apx:delta_thrd_11}, \eqref{apx:delta_thrd_2}, and \eqref{apx:delta_thrd_22}, it is clear that ${V}_\alpha(\boldsymbol{s}')|_ p\leq{V}_\alpha(\boldsymbol{s}')|_q$ for all $q\geq0$, i.e., $p$ is optimal for state $\boldsymbol{s}'$.

\subsubsection{The case of $p=0$, $\delta'<\delta$, and $d'=d$}
Since $p=0$ is optimal for state $\boldsymbol{s}$, we  have
\begin{align}
    \label{apx:delta_thrd_30}
                    V_\alpha(\boldsymbol{s}) &= \delta+1 + wd + \widebar{V}_\alpha(\boldsymbol{s}, 0) \\
    \label{apx:delta_thrd_300}
     &\leq C(\boldsymbol{s}, q) + \widebar{V}_\alpha(\boldsymbol{s}, q) = {V}_\alpha(\boldsymbol{s})|_ q
\end{align}
for any $q>0$,.

For any $q>0$, we also have
\begin{align}
{V}_\alpha(\boldsymbol{s}')|_ 0 &=\delta'+1 + wd + \widebar{V}_\alpha(\boldsymbol{s}', 0)\\
    \label{apx:delta_thrd_31}
    &<\delta+1 + \widebar{V}_\alpha(\boldsymbol{s}, 0) \\
    \label{apx:delta_thrd_32}
    & \leq C(\boldsymbol{s}, q) + wd + \widebar{V}_\alpha(\boldsymbol{s}, q) \\
    \label{apx:delta_thrd_33}
    & = C(\boldsymbol{s}', q) + \widebar{V}_\alpha(\boldsymbol{s}', q) = {V}_\alpha(\boldsymbol{s}')|_q,
\end{align}
in which \eqref{apx:delta_thrd_31} follows $\delta'<\delta$ and $\widebar{V}_\alpha(\boldsymbol{s}', 0) \leq\widebar{V}_\alpha(\boldsymbol{s}, 0)$ (since both $\widebar{V}_\alpha(\boldsymbol{s}, 0)$ and ${V}_\alpha(\boldsymbol{s})$ are non-decreasing with $\delta$);
    \eqref{apx:delta_thrd_32} follows \eqref{apx:delta_thrd_30} and \eqref{apx:delta_thrd_300};
 \eqref{apx:delta_thrd_33} follows $C(\boldsymbol{s}, q)=1+wD(q)=C(\boldsymbol{s}', q)$ and $\widebar{V}_\alpha(\boldsymbol{s}, q)=\widebar{V}_\alpha(\boldsymbol{s}', q)$ (since the destination state after action $P_k=q$ is independent from the AoI).

That is, $p=0$ is optimal for $\boldsymbol{s}'$.

\subsubsection{The case $p>0$, $\delta'=\delta$, and $d'>d$}
By using a positive transmit power $p>0$, the distortion of the next block is $D(p)$, which is independent from the distortion $d$ of the starting state.
    Thus, by start from either $\boldsymbol{s}$ or $\boldsymbol{s}'$, the system has the same potential state set for the next block, i.e., state $\boldsymbol{t}_1=(1,D(q), b-q+1)$ and  state $\boldsymbol{t}_0=(1,D(q), b-q)$.
    Hence, we have
\begin{align}\label{apx:d_thrd_1}
    {V}_\alpha(\boldsymbol{s})|_ p &= {V}_\alpha(\boldsymbol{s}')|_ p,\\
    \label{apx:d_thrd_2}
    {V}_\alpha(\boldsymbol{s})|_ q &= {V}_\alpha(\boldsymbol{s}')|_ q, ~\text{if}~q>0.
\end{align}

 For the case $q=0$, the system would transit from $\boldsymbol{s}$ to $\boldsymbol{t}_{01}=(\delta+1,d, b+1)$ or $\boldsymbol{t}_{00}=(\delta+1,d, b)$, and from $\boldsymbol{s}'$ to $\boldsymbol{t}'_{01}=(\delta+1,d', b+1)$ or $\boldsymbol{t}'_{00}=(\delta+1,d', b)$.
    Since ${V}_\alpha(\boldsymbol{s})$ is non-decreasing with $d$, we have $\widebar{V}_\alpha(\boldsymbol{s}, 0) \leq \widebar{V}_\alpha(\boldsymbol{s}', 0)$, and thus
\begin{align}\label{apx:d_thrd_3}
    {V}_\alpha(\boldsymbol{s})|_0 &= \delta+1+wd + \widebar{V}_\alpha(\boldsymbol{s}, 0) \\
    \label{apx:d_thrd_4}
            &\leq \delta+1+wd' + \widebar{V}_\alpha(\boldsymbol{s}', 0)
            = {V}_\alpha(\boldsymbol{s}')|_ 0.
\end{align}

By combing \eqref{apx:d_thrd_2} and \eqref{apx:d_thrd_3}--\eqref{apx:d_thrd_4}, we have
\begin{equation}\label{apx:d_thrd_5}
    {V}_\alpha(\boldsymbol{s})|_ q \leq {V}_\alpha(\boldsymbol{s}')|_ q, ~~\forall~ q\geq0.
\end{equation}

Moreover, $p>0$ is optimal for $\boldsymbol{s}$ implies that for any $q\neq p$,
\begin{align}\label{apx:d_thrd_6}
    V_\alpha(\boldsymbol{s}) = {V}_\alpha(\boldsymbol{s})|_ p\leq {V}_\alpha(\boldsymbol{s})|_ q.
\end{align}

By combing \eqref{apx:d_thrd_1}, \eqref{apx:d_thrd_5}, \eqref{apx:d_thrd_6}, we finally have
\begin{equation}\label{apx:d_thrd_8}
    {V}_\alpha(\boldsymbol{s}')|_ p = {V}_\alpha(\boldsymbol{s})|_ p
    \leq {V}_\alpha(\boldsymbol{s})|_ q  \leq {V}_\alpha(\boldsymbol{s}')|_ q,  ~\forall~ q\geq0.
\end{equation}

That is, $p$ is also optimal for state $\boldsymbol{s}'$.

\subsubsection{In the case $p=0$, $\delta'=\delta$, and $d'<d$}
As discussed in the previous sub-subsection, we have $\widebar{V}_\alpha(\boldsymbol{s}', 0) \leq \widebar{V}_\alpha(\boldsymbol{s}, 0)$ since $d'<d$.
    Also, we have  ${V}_\alpha(\boldsymbol{s}')|_ q = {V}_\alpha(\boldsymbol{s})|_q$ for any $q>0$.
Therefore, the following result holds.
\begin{equation}\label{apx:d_thrd_9}
    {V}_\alpha(\boldsymbol{s}')|_ 0 = {V}_\alpha(\boldsymbol{s})|_ 0
    \leq {V}_\alpha(\boldsymbol{s})|_ q  \leq {V}_\alpha(\boldsymbol{s}')|_ q,  \forall~ q>0
\end{equation}
which shows that $p=0$ is optimal for state $\boldsymbol{s}'$.
This completes the proof of \textit{Theorem} \ref{th:3_pproperty}.
\end{proof}


\small{
\bibliographystyle{IEEEtran}

\begin{thebibliography}{11}

\bibitem{Monitoring-IoTJ-2018}
F, Montori, L. Bedogni, and L. Bononi, ``A collaborative Internet of Things architecture for
smart cities and environmental monitoring," \textit{IEEE Internet Things J.}, vol. 5, no. 2, pp. 592--605, Apl. 2018.

\bibitem{industrial-auto}
L. Ascorti, S. Savazzi, G. Soatti, M. Nicoli, M. Sisinni, and S. Galimberti, ``A wireless cloud network platform for iIndustrial process automation: Critical data publishing and distributed sensing," \textit{IEEE Trans. Instrum. Meas.}, vol. 66, no. 4, pp. 592--603, Apr. 2017.

\bibitem{uav-tracking}
J. Gu, T. Su, Q. Wang, X. Du, and M. Guizani, ``Multiple moving targets surveillance based on a cooperative network for multi-UAV," \textit{IEEE Commun. Mag.}, vol. 56, no. 4, pp. 82--89, Apr. 2018.

\bibitem{dong-TWC-correlated_sensing}
Y. Dong, Z. Chen, J. Wang, and B. Shim, ``Optimal power control for transmitting correlated sources with energy harvesting constraints," \textit{IEEE Trans. Wireless Commun.,} vol.17, no. 1, pp. 461--476, Jan. 2018.

\bibitem{Dong-TSP-2019}
Y. Dong, ``Distributed sensing with orthogonal multiple access: To code or not to code?" sumitted to \textit{IEEE Trans. Signal Process.,} May, 2019.

\bibitem{Xiao-Linear-2008}
J.-J. Xiao, S. Cui, Z.-Q. Luo, and A. J. Goldsmith,   ``Linear coherent decentralized estimation," \textit{IEEE Trans. Signal Process.,} vol. 56, no. 2, pp. 757--770, Feb. 2008.

%

\bibitem{Vnet-1-2011}
S. Kaul, M. Gruteser, V. Rai, and J. Kenney, ``Minimizing age of information in vehicular networks," in \textit{Proc. IEEE SECON}, Salt Lake, Utah, USA, Jun. 2011, pp. 350--358.

\bibitem{Gu-2019-mornitoring}
Y. Gu, H. Chen, Y. Zhou, Y. Li,  and B. Vucetic, ``Timely status update in Internet of Things monitoring systems: An age-energy tradeoff," \textit{IEEE Internet Things J.}, vol. 6, no. 3, pp. 5324--5335, Jun. 2019.

\bibitem{Niu-2019-RR1}
Z. Jiang, B. Krishnamachari, X.  Zheng, S. Zhou, and Z. Niu,``Timely status update in wireless uplinks: Analytical solutions with asymptotic optimality," \textit{IEEE Internet Things J.}, vol. 6, no. 2, pp. 3885--3898, Apl. 2019.

\bibitem{health-Proc-2012}
I. Bisio, F. Lavagetto, M. Marchese, and A. Sciarrone, ``Smartphone-based user activity recognition method for health remote monitoring applications", in \textit{ Proc. Intl. Conf.  Pervasive  Embedded
Comput. Commun. Sys.}, Rome, Italy, Feb. 2012, pp. 200--205.

\bibitem{Gu-2019-cognitive}
Y. Gu, H. Chen, C. Zhai, Y.  Li, and B. Vucetic, ``Minimizing age of information in cognitive radio-based IoT systems: Underlay or overlay?" \textit{IEEE Internet Things J.}, vol. 99, no. 5, pp. 3885--3898, Oct. 2019.

\bibitem{Dong-2018-infocom}
Y. Dong, Z. Chen, and P. Fan, ``Uplink age of information of unilaterally powered two-way data exchanging systems,'' in \textit{Proc. IEEE Conf. Comput. Commun. Workshops (INFOCOM),} Honolulu, HI, USA, Apr. 2018, pp. 559--564.

\bibitem{Dong-2019-access}
Y. Dong, Z. Chen, and P.i Fan, ``Timely two-way data exchanging in unilaterally powered fog computing systems," \textit{IEEE Access}, vol. 7, pp. 21103--21117, Feb. 2019.

\bibitem{Dong-2019-jcn}
C. Hu and Y. Dong, ``Age of information of two-way data exchanging system with power-splitting," \textit{IEEE J. Commun. Netw.}, vol. 21, no. 3, pp. 295--306, Jun. 2019.

\bibitem{ulukus-JSAC-eh}
S. Ulukus et al., ``Energy harvesting wireless communications: A review of recent advances," \textit{IEEE J. Sel. Areas Commun.,} vol. 33, no. 3, pp. 360--381, Mar. 2015.

\bibitem{Uysal-2015-ita}
B. T. Bacinoglu, E. T. Ceran, and E. Uysal-Biyikoglu, ``Age of information under energy replenishment constraints," in \textit{Proc. IEEE Inf. Theory App. Wksp (ITA)}, San Diego, CA, USA, Feb. 2015, pp. 1--6.

\bibitem{Ulukus-TWC-2019}
A. Arafa and S.Ulukus, ``Timely updates in energy harvesting two-hop networks: Offline and online policies," \textit{IEEE Trans. Wireless Commun.,} vol. 18, no. 8, pp. 4017--4030, Aug. 2019.

\bibitem{Uysal-2017-ISIT}
B. T. Bacinoglu and E. Uysal-Biyikoglu, ``Scheduling status updates to minimize age of information with an energy harvesting sensor," in \textit{Proc. IEEE Int. Symp. Inf. Theory (ISIT)}, Aachen, Germany, Jun. 2017, pp. 1--6.

\bibitem{Yangjing-2017-TGCN}
X. Wu, J. Yang, and J. Wu, ``Optimal status update for age of information minimization with an energy harvesting source," \textit{IEEE Trans. Green Commun. and Netw.}, vol. 2, no. 1, pp. 193--204, Mar. 2018.

\bibitem{Zhou-TVT-2016}
H. Zhou, T. Jiang, C. Gong, and Y. Zhou, ``Optimal estimation in wireless sensor networks with energy harvesting,"  \textit{IEEE Trans. Veh. Technol.}, vol. 65, no. 11, pp.  9386--9396, Nov. 2016.

\bibitem{Reily-ComST-2014}
C. Reilly, A. Gluhak, M. A. Imran, S. Rajasegarar, ``Anomaly detection in wireless sensor networks in a non-stationary environment,"  \textit{IEEE Commun. Surveys \& Tutorials}, vol.  16, no. 3,  pp. 1413--1432, Jan. 2014.

\bibitem{Joda-TCom-2013}
R. Joda and Farshad Lahouti, ``Delay-limited source and channel coding of quasi-stationary sources over block fading
channels: Design and scaling laws," \textit{IEEE Trans. Commun.,} vol. 61, no. 4, pp. 1562--1572, Apl. 2013.

\bibitem{Cover-IT-Book}
T. M. Cover and J. A. Thomas, \textit{Elements of Information Theory,} Wiley,
New York,  2ed edition, 2012.

\bibitem{Yates-2012-age}
S. K. Kaul, R. D. Yates, and M. Gruteser, ``Real-time status: How often should one update?" in \textit{Proc. IEEE INFOCOM}, Orlando, FL, USA, Mar. 2012, pp. 2731--2735.

\bibitem{Ulukus-2012-awgn}
O. Ozel and S. Ulukus, ``Achieving AWGN capacity under stochastic energy harvesting," \textit{IEEE Trans.  Inform. Theory}, vol. 58, no. 10, pp. 6471--6483, Oct. 2012.

\bibitem{JYang-2012-jcn-backward}
J. Yang and S. Ulukus, ``Optimal packet scheduling in a multiple access channel with energy harvesting transmitters," \textit{J. Commun. Netw.}, vol. 14, no. 2, pp. 140--150, Apr. 2012.

\bibitem{Li-integeropt-2006}
D. Li and X. Sun, \textit{Nonlinear integer programming,} Springer,
New York,   2006.

\bibitem{Ross-MDP-1970}
S.M. Ross, \textit{Applied probability models with optimization applications,} CA: Holden-Day,
San Francisco, 1970.

\end{thebibliography}

}

\end{document}